\documentclass[10pt,a4paper]{article}
\usepackage[utf8]{inputenc}
\usepackage{amsmath, amsthm, amssymb}
\usepackage{enumerate}
\usepackage[svgnames]{xcolor}
\usepackage[toc,page]{appendix}
\usepackage[a4paper]{geometry}
\usepackage{tikz-cd}
\newtheorem{thm}{Theorem}[section]
\newtheorem{lem}[thm]{Lemma}
\newtheorem{tvrz}[thm]{Proposition}
\newtheorem{lemma}[thm]{Lemma}
\newtheorem{cor}[thm]{Corollary}
\newtheorem{theorem}[thm]{Theorem}
\theoremstyle{definition}
\newtheorem{definice}[thm]{Definition}
\theoremstyle{remark}
\newtheorem{rem}[thm]{Remark}
\theoremstyle{definition}
\newtheorem{example}[thm]{Example}

\def\J{\mathcal{J}}
\def\R{\mathbb{R}}

\def\T{\mathcal{T}}
\def\<{\langle}
\def\>{\rangle}
\def\~{\widetilde}
\def\^{\wedge}

\def\g{\mathfrak{g}}

\def\d{\mathfrak{d}}

\def\io{\mathit{i}}
\def\F{\mathcal{F}}
\def\G{\mathcal{G}}
\def\K{\mathcal{K}}

\def\H{\mathcal{H}}
\def\A{\mathcal{A}}
\def\D{\mathcal{D}}

\def\tC{\tilde{C}}

\def\fK{\mathbf{K}}
\def\gTM{\mathbb{T}M}

\def\gTP{\mathbb{T}P}

\def\gm{\mathbf{G}}

\def\RS{\mathcal{R}}

\def\fT{\mathbf{T}}

\def\fPsi{\mathbf{\Psi}}

\def\cD{\nabla}

\def\hcD{\widehat{\nabla}}
\def\hcDL{\widehat{\nabla}^{LC}}

\def\cDL{\nabla^{LC}}

\def\cif{C^{\infty}(M)}

\def\O{\mathcal{O}}

\newcommand{\bm}[4]{\begin{pmatrix} #1 & #2 \\ #3 & #4 \end{pmatrix}}

\newcommand{\vf}[1]{ \mathfrak{X}^{#1}(M)}
\newcommand{\df}[1]{ \Omega^{#1}(M)}

\newcommand{\Li}[1]{ \mathcal{L}_{#1}}

\DeclareMathOperator{\BlockDiag}{BlockDiag}

\DeclareMathOperator{\vol}{vol}
\DeclareMathOperator{\End}{End}
\DeclareMathOperator{\rank}{rank}
\DeclareMathOperator{\Hom}{Hom}
\DeclareMathOperator{\LC}{LC}

\DeclareMathOperator{\Aut}{Aut}

\DeclareMathOperator{\Sym}{Sym}
\DeclareMathOperator{\Ric}{Ric}

\DeclareMathOperator{\Div}{div}
\DeclareMathOperator{\hRic}{\widehat{R}ic}

\begin{document}
\begin{flushright}
\today
\end{flushright}
\vspace{0.7cm}
\begin{center}

\baselineskip=13pt {\Large \bf{Courant Algebroid Connections and String Effective Actions}\\}
 \vskip0.5cm
 {\it In memory of our friend Martin Doubek}  
 \vskip0.7cm
 {\large{ Branislav Jurčo$^{1}$, Jan Vysoký$^{2,3}$}}\\
 \vskip0.6cm
$^{1}$\textit{Mathematical Institute, Faculty of Mathematics and Physics,
Charles University\\ Prague 18675, Czech Republic, jurco@karlin.mff.cuni.cz}\\
\vskip0.3cm

$^{2}$\textit{Institute of Mathematics of the Czech Academy of Sciences \\ Žitná 25, Prague 11567, Czech Republic, vysoky@math.cas.cz}\\
\vskip0.3cm

$^{3}$\textit{Max Planck Institute for Mathematics\\
Vivatsgasse 7, Bonn 53113, vysokjan@mpim-bonn.mpg.de}\\
\vskip0.5cm
\end{center}

\begin{abstract}
Courant algebroids are a natural generalization of quadratic Lie algebras, appearing in various contexts in mathematical physics. A connection on a Courant algebroid gives an analogue of a covariant derivative compatible with a given fiber-wise metric. Imposing further conditions resembling standard Levi-Civita connections, one obtains a class of connections whose curvature tensor in certain cases gives a new geometrical description of equations of motion of low energy effective action of string theory. Two examples are given. One is the so called symplectic gravity, the second one is an application to the the so called heterotic reduction. All necessary definitions, propositions and theorems are given in a detailed and self-contained way. 
\end{abstract}


\section{Introduction}
In these lecture notes, we give a consistent and detailed introduction to an interesting application of generalized Riemannian geometry to bosonic string and to bosonic part of heterotic string. The main focus is on the equations of motion for the respective low energy effective actions. We extend and discuss in more detail the ideas sketched in our papers \cite{Jurco:2015xra} and \cite{Jurco:2015bfs}. In particular, here we use a different definition of the Riemann tensor, the one introduced in the double field theory by Hohm and Zwiebach \cite{Hohm:2012mf}. Further relevant references will be given in the following sections.

In Section \ref{sec_algebroids}, we provide a necessary introduction into the theory of Leibniz algebroids. Particular examples which are better known are e.g. Lie algebroid or Courant algebroids. The latter one are natural generalization of quadratic Lie algebras. 

Section \ref{sec_genmetric} introduces a generalization of Riemannian metric, which is compatible with the fiber-wise pairing on the Courant algebroid. In particular, we discuss various equivalent reformulations of this concept. 

Courant algebroid connections  naturally combine ordinary vector bundle connections with linear connections on manifolds. In particular, we discuss in detail definitions of suitable torsion and curvature operators. This is a main subject of Section \ref{sec_connections}. 

Assuming that we are in addition given a generalized metric, we may investigate compatible Courant algebroid connections. Moreover, imposing also the torsion-freeness condition, we can speak of generalized Levi-Civita connections, which we do in Section \ref{sec_LCconnections}. 

In particular, we attempt to classify generalized Levi-Civita connections. Also, we derive some important properties of those, which will prove useful in calculations of their Riemann curvature tensors. Complete answer to the classification problem can be given in case of exact Courant algebroids, cf.  Section \ref{sec_exact}. In other words, we find all Levi-Civita connections on a generalized tangent bundle and calculate their Ricci scalars. Moreover, we investigate the so called Ricci compatibility condition for such connections. 

It turns out that the Ricci compatibility condition plays, together with the flatness condition, a central role  in the geometrical description of conditions for vanishing of beta functions  as they are known in string theory. Equivalently, these conditions are equivalent to equations of motion for low energy bosonic string action. This is the main subject of  Section \ref{sec_EOM}. 

The observation mentioned in the above paragraph provides us with a quite useful mathematical tool. For example, we can use it to prove the classical equivalence of two at first glance unrelated field theories in Section \ref{sec_BIG}. One is the already mentioned low energy effective action of the bosonic string, the second one is the so called symplectic gravity. 

An another example is based on the reductions of Courant algebroids. We propose a suitable generalization of Kaluza-Klein reduction for low energy string actions, based on the paper \cite{Jurco:2015bfs} and briefly discussed in Section \ref{sec_KKR}. This is relevant to the heterotic string.
\section*{Conventions}
We assume that all manifolds are smooth, real, Hausdorff and locally compact. Vector bundles are real and have finite rank. By $\Gamma(E)$ we denote the module of global sections of a vector bundle $E$. Let $E$ and $E'$ be two vector bundles over $M$. Then $\Hom(E,E')$ denotes the set of of vector bundle morphisms from $E$ to $E'$ over an identity map on the base space. $\End(E) = \Hom(E,E)$. 

We use a slightly misleading notation $\Omega^{p}(E)$ for sections of $\Lambda^{p} E^{\ast}$ and even call them for simplicity $p$-forms on $E$. They \emph{are not} $p$-forms on the total space manifold $E$ (i.e., not sections of $\Lambda^{p} T^{\ast}E$). Similarly, $\T_{p}^{q}(E)$ denotes the module of $\cif$-multilinear maps from $p$ copies of $\Gamma(E)$ and $q$ copies of $\Gamma(E^{\ast})$ into $\cif$. Elements of $\T_{p}^{q}(E)$ are called tensors on $E$. 

Given a $2$-form $B \in \df{2}$, we often view it as a map $B \in \Hom(TM,T^{\ast}M)$ defined by inserting the vector field as its \emph{second argument}, $B(X) = B(\cdot,X) \in \df{1}$, for all $X \in \vf{}$. Note that we use the same symbols for the form and the corresponding map. The same convention is used for $2$-vector fields. 
\section{Leibniz, Lie and Courant algebroids} \label{sec_algebroids}
Let us start by recalling definitions of three kinds of algebroids appearing in this paper. The most general concept is the one of a Leibniz algebroid. Leibniz algebroids were first introduced in \cite{Loday1993} by Loday. In mathematics, they are usually called Loday algebroids and the definitions may vary according to subtleties included into axioms, see e.g. the introduction of \cite{2011arXiv1103.5852G} for a more detailed discussion. For our purposes, it is sufficient to think about a Leibniz algebroid as a Leibniz algebra on the module of sections of a vector bundle respecting to some extent the multiplication by a smooth function. 
\begin{definice}
Let $E$ be a vector bundle over a manifold $M$ and $\rho \in \Hom(E,TM)$ a smooth vector bundle morphism called the \textbf{anchor}. Further, let $[\cdot,\cdot]_{E}: \Gamma(E) \times \Gamma(E) \rightarrow \Gamma(E)$ be an $\R$-bilinear map. Then $(E,\rho,[\cdot,\cdot]_{E})$ is called a \textbf{Leibniz algebroid}, if 
\begin{equation} \label{eq_lebnizrule}
[\psi,f\psi']_{E} = f [\psi,\psi']_{E} + (\rho(\psi).f) \psi', 
\end{equation}
for all $\psi,\psi' \in \Gamma(E)$ and $f \in \cif$, and $(E,\Gamma(E))$ is a Leibniz algebra, that is 
\begin{equation} \label{eq_leibnizidentity}
[\psi,[\psi',\psi'']_{E}]_{E} = [[\psi,\psi']_{E},\psi'']_{E} + [\psi',[\psi,\psi'']_{E}]_{E}
\end{equation}
holds for all $\psi,\psi',\psi'' \in \Gamma(E)$. The condition (\ref{eq_lebnizrule}) is called the \textbf{Leibniz rule}, whereas (\ref{eq_leibnizidentity}) is called the \textbf{Leibniz identity}. 
\end{definice}
In general, the bracket is  not assumed to be skew-symmetric. In particular, there is no obvious Leibniz rule with respect to the left input of the bracket. Moreover, the order of brackets in (\ref{eq_leibnizidentity}) is important. It is thus practical to view the operator $[\psi,\cdot]_{E}$ as an inner derivation of the bracket $[\cdot,\cdot]_{E}$ itself. A combination of the two axioms (\ref{eq_lebnizrule}, \ref{eq_leibnizidentity}) immediately yields the following:
\begin{lemma} \label{lem_rhoishom}
Let $(E,\rho,[\cdot,\cdot]_{E})$ be a Leibniz algebroid. Then its anchor $\rho$ preserves the brackets:
\begin{equation} \label{eq_rhoishom}
\rho([\psi,\psi']_{E}) = [\rho(\psi),\rho(\psi')],
\end{equation}
for all $\psi,\psi' \in \Gamma(E)$, where the commutator  on the right-hand side is the vector field commutator. 
\end{lemma}
\begin{proof}
Use (\ref{eq_leibnizidentity}) on the triple $(\psi,\psi',f\psi'')$ and apply (\ref{eq_lebnizrule}) twice. 
\end{proof}
In fact, the property (\ref{eq_rhoishom}) can be viewed as a necessary condition for the consistence of (\ref{eq_leibnizidentity}) and (\ref{eq_lebnizrule}). For every Leibniz algebroid $(E,\rho,[\cdot,\cdot]_{E})$, one can extend, analogously to the Lie derivative, the bracket to an operator $\Li{}^{E}$ on the whole tensor algebra $\T(E)$ . In particular, define 
\begin{equation}
\Li{\psi}^{E}(f) = \rho(\psi).f, \; \; \Li{\psi}^{E}(\psi') = [\psi,\psi']_{E},
\end{equation}
for all $\psi \in \Gamma(E)$, $f \in \cif \equiv \T^{0}_{0}(E)$ and $\psi' \in \Gamma(E) \equiv \T^{1}_{0}(E)$. On $1$-forms, set 
\begin{equation}
\< \Li{\psi}^{E}(\eta), \psi'\> = \rho(\psi).\<\eta,\psi'\> - \<\eta, [\psi,\psi']_{E}\>,
\end{equation}
for all $\eta \in \T^{0}_{1}(E) \equiv \Omega^{1}(E) \equiv \Gamma(E^{\ast})$ and $\psi,\psi' \in \Gamma(E)$. Leibniz rule (\ref{eq_lebnizrule}) ensures that $\Li{\psi}^{E}(\eta) \in \Omega^{1}(E)$. Finally, its value on any tensor in $\T(E)$ is determined by usual tensor product rule:
\begin{equation}
\Li{\psi}^{E}( \tau \otimes \sigma) = \Li{\psi}^{E}(\tau) \otimes \sigma + \tau \otimes \Li{\psi}^{E}(\sigma), 
\end{equation}
for all $\psi \in \Gamma(E)$ and $\tau,\sigma \in \T(E)$. Leibniz identity (\ref{eq_leibnizidentity}) can be then used to prove that
\begin{equation}
\Li{[\psi,\psi']_{E}}^{E} = [ \Li{\psi}^{E}, \Li{\psi'}^{E}],
\end{equation}
for all $\psi,\psi' \in \Gamma(E)$. Moreover, the operator $\Li{\psi}^{E}$ restricts naturally on the exterior algebra $\Omega^{\bullet}(E)$, and one can show that also the usual formula
\begin{equation}
\io_{[\psi,\psi']_{E}} \omega = \Li{\psi}^{E}( \io_{\psi'}\omega) - \io_{\psi'}( \Li{\psi}^{E}\omega)
\end{equation}
holds for all $\omega \in \Omega^{\bullet}(E)$. However, in general, there is  no Leibniz algebroid analogue of the de Rham differential which could be used to obtain the full set of Cartan magic formulas. The only obstacle is the lacking skew-symmetry of the bracket, which is avoided in the more familiar case of a Lie algebroid. 
\begin{definice}
A Leibniz algebroid $(E,\rho,[\cdot,\cdot]_{E})$ with the bracket $[\cdot,\cdot]_{E}$  being skew-symmetric is called a \textbf{Lie algebroid}. The Leibniz identity (\ref{eq_leibnizidentity}) is then called \textbf{Jacobi identity} and $(\Gamma(E),[\cdot,\cdot]_{E})$ becomes an ordinary real Lie algebra. 
\end{definice}
In this case, one can define the differential $d^{E}$ on $\Omega^{\bullet}(E)$ inductively by imposing the Cartan formula
\begin{equation} \label{eq_cartanmagic}
\Li{\psi}^{E}(\omega) = d^{E}( \io_{\psi} \omega) + \io_{\psi}( d^{E}\omega), 
\end{equation}
for all $\omega \in \Omega^{\bullet}(E)$. The reason why this works is the existence of the Leibniz rule in the left input of the bracket $[\cdot,\cdot]_{E}$ due to its skew-symmetry. 

\begin{example}
Here are some classical examples of Lie algebroids. The most basic one is $(TM,1_{TM},[\cdot,\cdot])$, where the bracket is the vector field commutator. Further, for $M = \{ pt \}$, every Lie algebroid is an ordinary Lie algebra. Finally, let $(M,\Pi)$ be a Poisson manifold. One can view the Poisson bivector $\Pi \in \vf{2}$ as a vector bundle map $\Pi \in \Hom(T^{\ast}M,TM)$. Define the bracket $[\cdot,\cdot]_{\Pi}$ on $\Gamma(T^{\ast}M) \equiv \df{1}$ as 
\begin{equation}
[\xi,\eta]_{\Pi} = \Li{\Pi(\xi)} \eta - \io_{\Pi(\eta)} d\xi,
\end{equation}
for all $\xi,\eta \in \df{1}$. Then $(T^{\ast}M, \Pi, [\cdot,\cdot]_{\Pi})$ defines a Lie algebroid. Jacobi identities for $[\cdot,\cdot]_{\Pi}$ are equivalent to the vanishing of the Schouten-Nijenhuis bracket $[\Pi,\Pi]_{S} = 0$. 
\end{example}

Now, assume that one would like to generalize the concept of a quadratic Lie algebra. Instead of a one non-degenerate symmetric bilinear form, we have now a such one for each fiber of the vector bundle $E$, changing smoothly from fiber to fiber. 
\begin{definice}
Let $E$ be a vector bundle over a manifold $M$. Let $\<\cdot,\cdot\>_{E}: \Gamma(E) \times \Gamma(E) \rightarrow \cif$ be a $\cif$-bilinear and symmetric map. We say that $\<\cdot,\cdot\>_{E}$ is a \textbf{fiber-wise metric} on $E$ if the induced map $g_{E}: \Gamma(E) \rightarrow \Gamma(E^{\ast})$ defined as $\< g_{E}(\psi), \psi'\> := \<\psi,\psi'\>_{E}$ is a module isomorphism. Equivalently, $\<\psi,\psi'\>_{E} = 0$ for all $\psi' \in \Gamma(E)$ implies $\psi = 0$. 
\end{definice}

It follows from $\cif$-bilinearity of $\<\cdot,\cdot\>_{E}$ that a fiber-wise metric can be restricted onto each fiber of $E_{m}$ to endow it with a non-degenerate bilinear symmetric form smoothly  depending on the point $m \in M$. We will often use the notation $\<\cdot,\cdot\>_{E}$ and $g_{E}$ interchangeably. 

Now assume that $(E,\rho,[\cdot,\cdot]_{E})$ is a Lie algebroid, equipped with a fiber-wise metric $\<\cdot,\cdot\>_{E}$. The straightforward generalization of the concept of an invariant form leads to the requirement
\begin{equation} \label{eq_invariance}
\rho(\psi).\<\psi',\psi''\>_{E} = \< [\psi,\psi']_{E}, \psi''\>_{E} + \<\psi', [\psi,\psi'']_{E}\>_{E},
\end{equation}
for all $\psi,\psi',\psi'' \in \Gamma(E)$. Equivalently, one can write $\Li{\psi}^{E}(g_{E}) = 0$ for all $\psi \in \Gamma(E)$. However, one can immediately observe that this condition is not $\cif$-linear in $\psi$, and leads to very restrictive implications. Namely that $\rho = 0$, an utterly boring case. 
This follows for example from (\ref{eq_cartanmagic}) and non-degeneracy of $g_{E}$. There are two ways around this issue. One can weaken the condition (\ref{eq_invariance}), assuming that it holds just for $\psi',\psi'' \in \Gamma(\ker(\rho))$. This leads to a concept of a quadratic Lie algebroid, see \cite{chen2013}. The other solution is to drop the skew-symmetry of the bracket $[\cdot,\cdot]_{E}$, i.e. to revert to the more general concept of a Leibniz algebroid. This leads to a definition of a Courant algebroid. 

\begin{definice} \label{def_Courant}
Let $(E,\rho,[\cdot,\cdot]_{E})$ be a Leibniz algebroid and $\<\cdot,\cdot\>_{E}$  a fiber-wise metric on $E$. Let $\D: \cif \rightarrow \Gamma(E)$ be the (unique) map completing the commutative diagram
\begin{equation} \label{eq_Ddiagram}
\begin{tikzcd}
\cif \arrow[r,"\D", dashed] \arrow[d,"d"] & E \arrow[d,"g_{E}"] \\
T^{\ast}M \arrow[r,"\rho^{T}"] & E^{\ast} 
\end{tikzcd}
\end{equation}
Then $(E,\rho,\<\cdot,\cdot\>_{E},[\cdot,\cdot]_{E})$ is called a \textbf{Courant algebroid} if $\Li{\psi}^{E}(g_{E}) = 0$, that is 
\begin{equation} \label{eq_invarianceznovu}
\rho(\psi).\<\psi',\psi''\>_{E} = \< [\psi,\psi']_{E},\psi''\>_{E} + \<\psi', [\psi,\psi'']_{E} \>_{E},
\end{equation}
for all $\psi,\psi',\psi'' \in \Gamma(E)$ and for the symmetric part of the bracket $[\cdot,\cdot]_{E}$
\begin{equation} \label{eq_courantsympart}
[\psi,\psi']_{E} + [\psi',\psi]_{E} = \D{\<\psi,\psi'\>_{E}}
\end{equation}
holds for all $\psi,\psi' \in \Gamma(E)$. As a consequence, we have the left Leibniz rule in the form
\begin{equation} \label{eq_leftleibniz}
[f\psi,\psi']_{E} = f [\psi,\psi']_{E} - (\rho(\psi').f) \psi + \<\psi,\psi'\>_{E} \D{f}, 
\end{equation}
for all $\psi,\psi' \in \Gamma(E)$ and $f \in \cif$. 
\end{definice}
Observe that the map $\D$ can equivalently be defined by the equation $\<\D{f},\psi\>_{E} = \rho(\psi)f$, for all $\psi \in \Gamma(E)$ and $f \in \cif$. The axiom (\ref{eq_courantsympart}) is thus sometimes rewritten in the form 
\begin{equation} \label{eq_couraxjinak}
\< [\psi,\psi]_{E}, \psi'\>_{E} = \frac{1}{2} \rho(\psi').\<\psi,\psi\>_{E}. 
\end{equation}
We encourage the reader to show that (\ref{eq_leftleibniz}) ensures that (\ref{eq_invarianceznovu}) is $\cif$-linear in $\psi$, and thus avoids the issue discussed below equation (\ref{eq_invariance}).  We summarize some of the direct consequences of the above definition in the following lemma.
\begin{lem} \label{lem_courantprops}
Let $(E,\rho,\<\cdot,\cdot\>_{E},[\cdot,\cdot]_{E})$ be a Courant algebroid. Let $\rho^{\ast} \in \Hom(T^{\ast}M,E)$ be the map $\rho^{\ast} := g_{E}^{-1} \circ \rho^{T}$. Then there is the sequence
\begin{equation} \label{eq_courantsequence}
\begin{tikzcd}
0 \arrow[r] & T^{\ast}M \arrow[r,"\rho^{\ast}"] & E \arrow[r,"\rho"] & TM \arrow[r] & 0
\end{tikzcd}
\end{equation}
with $\rho \circ \rho^{\ast} = 0$. Moreover, the following equations hold:
\begin{equation} \label{eq_courantprops}
\< \D{f}, \D{g} \>_{E} = 0, \; \; [\D{f},\psi]_{E} = 0,
\end{equation}
for all $f,g \in \cif$ and $\psi \in \Gamma(E)$. 
\end{lem}
\begin{proof}
One can prove $\rho \circ \D = 0$ by applying $\rho$ on both sides of (\ref{eq_leftleibniz}). As the image of $\D$ generates the image of $\rho^{\ast}$, one gets $\rho \circ \rho^{\ast} = 0$. The remaining assertions can be obtained similarly. 
\end{proof}
Before proceeding further, let us recall some terminology. 
\begin{definice}
Let $(E,\rho,\<\cdot,\cdot\>_{E},[\cdot,\cdot]_{E})$ and $(E',\rho',\<\cdot,\cdot\>_{E'}, [\cdot,\cdot]_{E'})$ be two Courant algebroids over the same base manifold $M$. Let $\F \in \Hom(E,E')$ be a morphism of the underlying vector bundles. We say that $\F$ is a \textbf{morphism of Courant algebroids}, if it preserves all involved structures, that is
\begin{equation} \label{eq_iso}
\rho = \rho' \circ \F, \; \; \<\psi,\psi'\>_{E} = \<\F(\psi),\F(\psi')\>_{E'}, \; \; \F([\psi,\psi']_{E}) = [\F(\psi),\F(\psi')]_{E'}, 
\end{equation}
for all $\psi,\psi' \in \Gamma(E)$. If $\F$ is invertible, it is called an \textbf{isomorphism of Courant algebroids}.
\end{definice}
As usual in linear algebra, the inverse $\F^{-1}$ of a Courant algebroid morphism is automatically a Courant algebroid morphism. Courant algebroids are usually distinguished by various properties of their anchor, namely:
\begin{definice}
Let $(E,\rho,\<\cdot,\cdot\>_{E},[\cdot,\cdot]_{E})$ be a Courant algebroid. Then it is called \textbf{regular} if the anchor $\rho$ has a constant rank. It is called \textbf{transitive} if the anchor $\rho$ is fiber-wise surjective. It is called \textbf{exact} if the sequence (\ref{eq_courantsequence}) is exact. 
\end{definice}
Exact Courant algebroids over a fixed base space $M$ are classified in terms of the de Rham cohomology $H^{3}_{dR}(M)$. This is due to Ševera \cite{severaletters}. Every exact Courant algebroid is in fact isomorphic to the one which is presented in the following example.
\begin{example} \label{ex_Hdorfman}
Let $M$ be a manifold, and let $[H] \in \df{3}$ be represented a closed $3$-form $H$. Let $E = \gTM := TM \oplus T^{\ast}M$ be a \textbf{generalized tangent bundle} on $M$. Set $\rho := pr_{TM}$, the projection onto the first component of the direct sum. The fiber-wise metric $\<\cdot,\cdot\>_{E}$ is defined as 
\begin{equation} \label{eq_canpairing}
\< (X,\xi), (Y,\eta) \>_{E} = \eta(X) + \xi(Y), 
\end{equation}
for all $X,Y \in \vf{}$ and $\xi,\eta \in \df{1}$. Finally, set $[\cdot,\cdot]_{E}$ to be the \textbf{$H$-twisted Dorfman bracket $[\cdot,\cdot]_{D}^{H}$} which has the form
\begin{equation}
[(X,\xi),(Y,\eta)]^{H}_{D} = \big( [X,Y], \Li{X}\eta - \io_{Y}d\xi - H(X,Y,\cdot)\big),
\end{equation}
for all $(X,\xi),(Y,\eta) \in \Gamma(\gTM)$. All axioms are straightforward to verify. The only one which requires some work is the Leibniz identity (\ref{eq_leibnizidentity}), leading eventually to the requirement $dH = 0$. Note that for any $B \in \df{2}$, there is an isomorphism $[\cdot,\cdot]_{D}^{H} \approx [\cdot,\cdot]_{D}^{H + dB}$, which explains the classification using the de Rham cohomology classes. 
\end{example}
\begin{rem} \label{rem_precourant}
In fact, in most of what follows, one can consider a slightly more general concept of pre-Courant algebroids. First, consider a map $\J: \Gamma(E)^{\otimes 3} \rightarrow \Gamma(E)$, called usually a \textbf{Jacobiator}\footnote{Although, it should  probably be called a "Leibnizator" or "Lodayator".}, which measures the failure of the Leibniz identity (\ref{eq_leibnizidentity}):
\begin{equation} \label{eq_Jacobiator}
\J(\psi,\psi',\psi'') = [\psi,[\psi',\psi'']_{E}]_{E} - [[\psi,\psi']_{E},\psi'']_{E} - [\psi',[\psi,\psi'']_{E}]_{E}. 
\end{equation}
Now assume that $(E,\rho,\<\cdot,\cdot\>_{E},[\cdot,\cdot]_{E})$ is an algebraic structure satisfying some of the axioms of Courant algebroid, namely (\ref{eq_lebnizrule}) and (\ref{eq_courantsympart}). Moreover, \emph{assume} that $\rho$ is a bracket morphism, that is (\ref{eq_rhoishom}) holds. One can then show that the map (\ref{eq_Jacobiator}) is completely skew-symmetric and $\cif$-linear in all inputs, $\J \in \Omega^{3}(E) \otimes \Gamma(E)$. In fact, defining $\J'(\psi,\psi',\psi'',\psi''') = \<\J(\psi,\psi',\psi''),\psi'''\>_{E}$, one finds that $\J' \in \Omega^{4}(E)$. It thus makes sense to consider the following definition. One says that $(E,\rho,\<\cdot,\cdot\>_{E},[\cdot,\cdot]_{E})$ is a \textbf{pre-Courant algebroid}, if there exists a closed $4$-form $C \in \df{4}$, such that 
\begin{equation}
\< \J(\psi,\psi',\psi''), \psi''' \>_{E} = C(\rho(\psi),\rho(\psi'),\rho(\psi''),\rho(\psi''')),
\end{equation}
for all $\psi,\psi',\psi'',\psi''' \in \Gamma(E)$. For more details, see \cite{2012arXiv1205.5898L}. 
\end{rem}
Courant algebroids have a very interesting history, as it took some time to formulate their modern definition presented here. In particular, the original version used the skew-symmetrized version of the algebroid bracket, with a significantly more complicated set of axioms. We recommend the excellent historical summary \cite{2013SIGMA...9..014K} by Kosmann-Schwarzbach on this subject. 
\section{Generalized Riemmanian metric} \label{sec_genmetric}
One of the reasons why the geometry of the generalized tangent bundle $\gTM$ attracted the attention of both mathematicians and physicists is its power to unify the description of various objects known from the usual differential geometry. One of such concepts is the one of a generalized Riemannian metric defined for any vector bundle with a fiber-wise metric. At first, the definition might seem rather puzzling, but the reason for its name will become more clear in the case of a generalized tangent bundle. In its present form, it first appeared in the thesis \cite{Gualtieri:2003dx} of Gualtieri in relation to generalized almost complex structures. 

\begin{definice} \label{def_genmet1}
Let $(E,\<\cdot,\cdot\>_{E})$ be a vector bundle with a fiber-wise metric $\<\cdot,\cdot\>_{E}$. A \textbf{generalized Riemannian metric} is an automorphism $\tau \in \Aut(E)$, such that $\tau^{2} = 1$, and the formula
\begin{equation}
\gm(\psi,\psi') := \< \psi, \tau(\psi')\>_{E},
\end{equation}
for all $\psi,\psi' \in \Gamma(E)$, defines a \emph{positive definite} fiber-wise metric $\gm$ on $E$. 
\end{definice}
In order to simplify the writing, we will drop the word Riemannian in what follows. First, observe that, as a direct consequence of the definition, $\tau$ has to be both symmetric and orthogonal with respect to $\<\cdot,\cdot\>_{E}$. Next, note that the fiber-wise metric $\<\cdot,\cdot\>_{E}$ has always a locally constant signature. This is a side effect of its smoothness, see \cite{2013arXiv1307.2171D} for an explanation. To avoid an unnecessary discussion for each connected component of $M$, we may assume that $M$ is  connected, in which case we may introduce an equivalent definition of the generalized metric:
\begin{definice} \label{def_genmet2} 
Let $(E,\<\cdot,\cdot\>_{E})$ be a vector bundle with a fiber-wise metric $\<\cdot,\cdot\>_{E}$. A \textbf{generalized metric} is a subbundle $V_{+} \subseteq E$ which is a maximal positive definite subbundle of $E$ with respect to $\<\cdot,\cdot\>_{E}$. In other words, $V_{+}$ is positive definite and not a proper subbundle of some other positive definite subbundle. 
\end{definice}
As it is far from obvious that both definitions lead to the same objects, we discuss this in the proof of the following proposition. 
\begin{tvrz}
Definition \ref{def_genmet1} and Definition \ref{def_genmet2} describe the same object. 
\end{tvrz}
\begin{proof}
Let $\tau \in \Aut(E)$ be a generalized metric according to Definition \ref{def_genmet1}. At each fiber $E_{m}$, the induced map $\tau_{m} \in \Aut(E_{m})$ is an involution, hence a diagonalizable map. Let $V_{m \pm}$ denote the corresponding $\pm 1$ eigenspaces. Then $E_{m} = V_{m+} \oplus V_{m-}$. Obviously $V_{m+}$ is positive definite subspace with respect to the restriction $\<\cdot,\cdot\>_{E_{m}}$ which has a constant signature $(p,q)$ for all $m \in M$. Hence $\dim(V_{m+}) = p$, and similarly $\dim(V_{m-}) = q$. But one can write $V_{m+} = \ker(\tau - 1)_{m}$, which proves that $V_{+} := \ker( \tau - 1)$ is a maximal positive definite subbundle of $E$, and fits into Definition \ref{def_genmet2}. Note that by construction $E = V_{+} \oplus V_{-}$. Moreover, it is easy to see that $V_{-} = V_{+}^{\perp}$ and it is a maximal negative definite subbundle of $E$ with respect to $\<\cdot,\cdot\>_{E}$. 

Conversely, let $V_{+} \subseteq E$ be a maximal positive definite subbundle with respect to $\<\cdot,\cdot\>_{E}$. Set $V_{-} = V_{+}^{\perp}$. Clearly $V_{+} \cap V_{-} = 0$, as $V_{+}$ cannot contain any isotropic elements. Moreover, $V_{-}$ is negative definite, as any section of $V_{-}$ can  neither positive (contradiction with maximality) nor isotropic (since $V_{-} \cap V_{-}^{\perp} = V_{-} \cap V_{+} = 0$. Set $\tau(\psi_{+} + \psi_{-}) = \psi_{+} - \psi_{-}$ for all $\psi_{\pm} \in \Gamma(V_{\pm})$). It is easy to check that $\tau$ fits in Definition \ref{def_genmet1}. 
\end{proof}
Finally, there is a definition which is the one justifying the name of generalized metric. 
\begin{definice} \label{def_genmet3}
Let $\gm$ be a positive definite fiber-wise metric on $(E,\<\cdot,\cdot\>_{E})$. We say that $\gm$ is a \textbf{generalized metric} if the induced isomorphism $\gm \in \Hom(E,E^{\ast})$ defines a map orthogonal with respect to the fiber-wise metric $g_{E} = \<\cdot,\cdot\>_{E}$ on $E$ and the fiber-wise metric $g_{E}^{-1}$ on $E^{\ast}$:
\begin{equation}
g_{E}(\psi,\psi') = g_{E}^{-1}( \gm(\psi), \gm(\psi')), 
\end{equation}
for all $\psi,\psi' \in \Gamma(E)$. 
\end{definice}
It is easy to see that $\gm$ from Definition \ref{def_genmet2} and \ref{def_genmet3} are equivalent. We will use the word generalized metric interchangeably for all three kinds of objects, the involution $\tau$, the fiber-wise metric $\gm$, and a maximal positive definite subbundle $V_{+}$. 

Let $O(E) \subseteq \Aut(E)$ denote the group of orthogonal automorphisms with respect to $\<\cdot,\cdot\>_{E}$, that is every $\O \in O(E)$ satisfies $\<O(\psi),\O(\psi')\>_{E} = \<\psi,\psi'\>_{E}$ for all $\psi,\psi' \in \Gamma(E)$. There is a natural action of $O(E)$ on the space of generalized metrics. Namely, set
\begin{equation}
\gm'(\psi,\psi') = \gm(\O(\psi),\O(\psi')), 
\end{equation}
for all $\psi,\psi' \in \Gamma(E)$. Other quantities transform accordingly, namely $\tau' = \O^{-1} \tau \O$ and $V'_{\pm} = \O^{-1}(V_{\pm})$. It is easy to see that $\gm'$, $\tau'$ and $V'_{+}$ again define a generalized metric. 

On the generalized tangent bundle, the generalized metric has a form more familiar to physicists. 
\begin{tvrz}
Let $E = \gTM$ and $\<\cdot,\cdot\>_{E}$ be the fiber-wise metric defined by (\ref{eq_canpairing}). Let $\gm$ be a generalized metric on $E$. Then $\gm$ defines a unique pair $(g,B)$, where $g > 0 $ is a Riemannian metric on $M$ and $B \in \df{2}$. Conversely, any pair $(g,B)$ defines a unique generalized metric. 
\end{tvrz}
\begin{proof} \label{tvrz_gmasgandB}
Let $\gm$ be a generalized metric. It thus defines a positive definite subbundle $V_{+} \subseteq \gTM$ of rank $n$, as the signature of $\<\cdot,\cdot\>_{E}$ is $(n,n)$. As both $TM$ and $T^{\ast}M$ are isotropic, we have $V_{+} \cap TM = V_{-} \cap T^{\ast}M = 0$. This implies that $V_{+}$ is a graph of a vector bundle isomorphism $\A \in \Hom(TM,T^{\ast}M)$. We can decompose it as $\A = g + B$, where $g$ is a symmetric form on $M$ and $B \in \df{2}$. As $V_{+}$ is positive definite, it follows that $g > 0$ is a Riemannian metric. Note that $V_{-}$ is the graph of the map $\A' = -g + B$. The corresponding fiber-wise metric $\gm$ can be written in a formal block form 
\begin{equation} \label{eq_gmblock}
\gm = \bm{g- Bg^{-1}B}{Bg^{-1}}{-g^{-1}B}{g^{-1}}. 
\end{equation}
Conversely, given a pair $(g,B)$, one can define $\gm$ by (\ref{eq_gmblock}) and show that it satisfies the properties of Definition \ref{def_genmet3}. 
\end{proof}
Note that (\ref{eq_gmblock}) is the form appearing in physics. For example, its inverse appears naturally in the Hamiltonian density of the Polyakov string in the target space with backgrounds $(g,B)$. 

Given a Riemannian metric $g$, one can always define a generalized metric by $\G = \BlockDiag(g,g^{-1})$. 
With a $2$-form $B \in \df{2}$, one can associate the map $e^{B} \in \End(\gTM)$ defined as 
\begin{equation} \label{eq_etoB}
e^{B}(X,\xi) = (X, \xi + B(X)), 
\end{equation}
for all $(X,\xi) \in \Gamma(\gTM)$. One has $e^{B} \in O(E)$. Then the generalized metric $\gm$ can be written as 
\begin{equation} \label{eq_gmasproduct}
\gm = (e^{-B})^{T} \G e^{-B}.
\end{equation}
For the future reference, define a vector bundle isomorphisms $\fPsi_{\pm}$ mapping any vector field $X \in \vf{}$ onto its images in graphs $V_{\pm}$, respectively. Explicitly, one has $\fPsi_{\pm}(X) = (X, (\pm g + B)(X))$. 

To finish this section, assume that we have a Courant algebroid $(E,\rho,\<\cdot,\cdot\>_{E},[\cdot,\cdot]_{E})$. As before put $\rho^{\ast} = g_{E}^{-1} \circ \rho^{T}$. Let $\gm$ be a generalized metric on $E$. Define a symmetric bilinear form $h_{\gm}$ on $T^{\ast}M$ as 
\begin{equation} \label{eq_hgm}
h_{\gm}(\xi,\eta) = \gm( \rho^{\ast}(\xi), \rho^{\ast}(\eta)),
\end{equation}
for all $\xi,\eta \in \df{1}$. It follows that it is positive semidefinite. For transitive Courant algebroids, $h_{\gm}$ defines a fiber-wise metric on the cotangent bundle $T^{\ast}M$. One of the interesting features is that $h_{\gm}$ is preserved by Courant algebroid isomorphisms orthogonal with respect to respective generalized metrics. 

\begin{lemma} \label{lem_hgcovariant}
Let $(E,\rho,\<\cdot,\cdot\>_{E})$ and $(E',\rho',\<\cdot,\cdot\>_{E'})$ be a pair of vector bundles with anchors and fiber-wise metrics (brackets are not important for this lemma) and let $\F \in \Hom(E,E')$ be an isomorphism in the sense of (\ref{eq_iso}), excluding the brackets. Assume that $\gm$ is a generalized metric on $E$ and $\gm'$ a generalized metric on $E'$, such that 
\begin{equation} \label{eq_gmisometry}
\gm(\psi,\psi') = \gm'(\F(\psi),\F(\psi')), 
\end{equation}
for all $\psi,\psi' \in \Gamma(E)$. Then $h_{\gm} = h_{\gm'}$. 
\end{lemma}
\begin{proof}
Direct calculation.
\end{proof}
In the example of $E = \gTM$ and the generalized metric (\ref{eq_gmblock}), the symmetric form $h_{\gm}$ is trivial to calculate, as $\rho^{\ast}(\xi) = (0,\xi)$, and consequently $h_{\gm} = g^{-1}$. As $\rho$ is surjective, $h_{\gm}$ is positive definite, as expected. Observe that the only automorphisms of $(E,\rho,\<\cdot,\cdot\>_{E})$ are exactly of the form (\ref{eq_etoB}), and it follows from (\ref{eq_gmasproduct}) that $h_{\gm}$ is indeed invariant under such automorphisms. 
\section{Courant algebroid connections}\label{sec_connections}
Having a vector bundle $E$ with an anchor $\rho \in \Hom(E,TM)$, it is natural to consider on $E$ an obvious generalization of the ordinary $TM$-connections. For Courant algebroids connections were introduced in \cite{2007arXiv0710.2719G} or in unpublished notes \cite{alekseevxu}. 
\begin{definice}
Let $(E,\rho,\<\cdot,\cdot\>_{E},[\cdot,\cdot]_{E})$ be a Courant algebroid. Let $\cD: \Gamma(E) \times \Gamma(E) \rightarrow \Gamma(E)$ be an $\R$-bilinear map. We use the usual notation $\cD(\psi,\psi') \equiv \cD_{\psi}\psi'$, for all $\psi,\psi' \in \Gamma(E)$. We say that $\cD$ is a \textbf{Courant algebroid connection} if it satisfies:
\begin{equation} \label{eq_connleib}
\cD_{f \psi} \psi' = f \cD_{\psi} \psi', \; \; \cD_{\psi}(f\psi') = f \cD_{\psi}\psi' + (\rho(\psi).f) \psi', 
\end{equation}
for all $\psi,\psi' \in \Gamma(E)$ and $f \in \cif$, and it is compatible with the fiber-wise  metric $\<\cdot,\cdot\>_{E}$:
\begin{equation} \label{eq_conncomp}
\rho(\psi).\<\psi',\psi''\>_{E} = \< \cD_{\psi}\psi', \psi''\>_{E} + \<\psi', \cD_{\psi}\psi''\>_{E}, 
\end{equation}
for all $\psi,\psi',\psi'' \in \Gamma(E)$. 
\end{definice}
\begin{rem}
Recall that for every vector bundle $E$, there exists a vector bundle $\mathfrak{D}(E)$, whose sections are called \emph{derivations} on $E$. To be more specific, $D \in \Gamma(\mathfrak{D}(E))$, if $D$ is an $\R$-linear map $D: \Gamma(E) \rightarrow \Gamma(E)$, such that there exists a vector field $a(D) \in \vf{}$ and
\begin{equation}
D(f\psi) = f D(\psi) + (a(D).f) \psi.
\end{equation}
One can prove that $a$ can be viewed as a surjective vector bundle morphism $a \in \Hom(\mathfrak{D}(E),TM)$, and the bracket $[D,D']_{\mathfrak{D}(E)}(\psi) := D(D'(\psi)) - D'(D(\psi))$ makes the triple $(\mathfrak{D}(E), a, [\cdot,\cdot]_{\mathfrak{D}(E)})$ into a transitive Lie algebroid, fitting into the short exact sequence of Lie algebroids:
\begin{equation}
\begin{tikzcd} \label{eq_derivaceSES}
0 \arrow[r]& \End(E) \arrow[r, hookrightarrow] & \mathfrak{D}(E) \arrow[r,"a"] & TM \arrow[r] & 0 
\end{tikzcd}
\end{equation}
Splittings of this sequence in the category of vector bundles correspond to ordinary vector bundle connections on $E$, in the category of Lie algebroids they correspond to flat vector bundle connections. Also, note that \ref{eq_derivaceSES} is an Atiyah sequence for the frame bundle of $E$. See the classical book \cite{Mackenzie} for a more detailed explanation. As $E$ is equipped with a fiber-wise metric $\<\cdot,\cdot\>_{E}$, one can define a following submodule of the module of derivations: 
\begin{equation}
\Sym(E) = \{ D \in \Gamma(\mathfrak{D}(E)) \; | \; \< D(\psi),\psi'\>_{E} + \<\psi, D(\psi')\>_{E} = a(D).\<\psi,\psi'\>_{E} \}
\end{equation}
The definition of the Courant algebroid connection can be now reformulated in this alternative language. Namely, any vector bundle map $\cD \in \Hom(E,\mathfrak{D}(E))$ is a Courant algebroid connection, if it takes values in the submodule $\Sym(E)$ and it fits into a diagram 
\begin{equation}
\begin{tikzcd}
E \arrow[r,"\cD"] \arrow[rd,"\rho"'] & \mathfrak{D}(E) \arrow[d,"a"]\\
& TM
\end{tikzcd}.
\end{equation}
This viewpoint gives us an immediate answer to the question of existence of Courant algebroid connections. We only need a splitting $\sigma \in \Hom(TM,\mathfrak{D}(E))$ of (\ref{eq_derivaceSES}) valued in $\Sym(E)$. But this is a vector bundle connection on $E$ compatible with $\<\cdot,\cdot\>_{E}$, which exists for any fiber-wise metric. Setting $\cD = \sigma \circ \rho$ gives the required Courant algebroid connection on $E$. 
\end{rem}
\begin{rem}
Having a Courant algebroid connection $\cD$, one can extend every operator $\cD_{\psi}$ to the whole tensor algebra $\T(E)$ using the conventional formulas. For example, when $\eta \in \Gamma(E^{\ast})$ is a $1$-form on $E$, one defines 
\begin{equation}
\< \cD_{\psi}\eta, \psi'\> = \rho(\psi).\<\eta,\psi'\> - \<\eta, \cD_{\psi}\psi'\>,
\end{equation}
for all $\psi,\psi' \in \Gamma(E)$. It follows from (\ref{eq_connleib}) that $\cD_{\psi}\eta \in \Gamma(E^{\ast})$. Note that the extension of $\cD$ to $\T(E)$ will always be  denoted by the same symbol. The metric compatibility (\ref{eq_conncomp}) can  now be rewritten as $\cD_{\psi} g_{E} = 0$ for all $\psi \in \Gamma(E)$, or simply as $\cD g_{E} = 0$. 
\end{rem}
Any Courant algebroid connection $\cD$ naturally defines a \textbf{covariant divergence} operator $\Div_{\cD}: \Gamma(E) \rightarrow \cif$. Indeed, set
\begin{equation}
\Div_{\cD}(\psi) = \< \cD_{\psi_{\lambda}} \psi, \psi^{\lambda} \>, 
\end{equation}
for all $\psi \in \Gamma(E)$, where $\{ \psi_{\lambda} \}_{\lambda=1}^{\rank(E)}$ is an arbitrary local frame on $E$, and $\{ \psi^{\lambda} \}_{\lambda=1}^{\rank(E)}$ is the respective dual frame on $E^{\ast}$. As $\cD_{\psi}$ is $\cif$-linear in $\psi$, it is well defined. One can easily derive the following Leibniz rule for this operator:
\begin{equation} \label{eq_leibnizdiv}
\Div_{\cD}(f\psi) = f \Div_{\cD}(\psi) + \rho(\psi).f,
\end{equation}
for all $\psi \in \Gamma(E)$ and $f \in \cif$. Now, recall the $\cif$-linear map $\D$ (\ref{eq_Ddiagram}) going in the opposite direction. It is thus interesting to investigate its compositions with the covariant divergence. 
\begin{lemma} \label{lem_XcDisavf}
Let $X_{\cD}: \cif \rightarrow \cif$ be an $\R$-linear operator defined as a composition 
\begin{equation} \label{eq_charvf}
X_{\cD} = \Div_{\cD} \circ \; \D.
\end{equation}
Then $X_{\cD}$ is a vector field on $M$. Moreover, this vector field is invariant under Courant algebroid isomorphisms (\ref{eq_iso}) in the sense explained in the proof. We call $X_{\cD} \in \vf{}$ the \textbf{characteristic vector field} of the Courant algebroid connection $\cD$. 
\end{lemma}
\begin{proof}
First, one can rewrite (\ref{eq_leibnizdiv}) using $\D$ and $\<\cdot,\cdot\>_{E}$ as 
\begin{equation}
\Div_{\cD}(f\psi) = f \Div_{\cD}(\psi) + \< \D{f}, \psi\>_{E}, 
\end{equation}
for all $\psi \in \Gamma(E)$ and $f \in \cif$. Moreover, note that $\D$ satisfies the Leibniz rule in the same form as a usual differential, $\D{(fg)} = \D{(f)} g  + f \D{(g)}$. Consequently 
\begin{equation}
\Div_{\cD}(\D{(fg)}) = \Div_{\cD}( \D{(f)} g + f \D{(g)} ) = \Div_{\cD}(\D{f}) g + f \Div_{\cD}(\D{g}) + 2 \< \D{f}, \D{g} \>_{E},
\end{equation}
for all $f,g \in \cif$. But recall that every Courant algebroid satisfies (\ref{eq_courantprops}). Hence the last term vanishes and we obtain the rule 
\begin{equation}
X_{\cD}(fg) = X_{\cD}(f) g + f X_{\cD}(g), 
\end{equation}
proving that $X_{\cD}$ is indeed a vector field on $M$. Next, let $(E,\rho,\<\cdot,\cdot\>_{E},[\cdot,\cdot]_{E})$ and $(E',\rho',\<\cdot,\cdot\>_{E'},[\cdot,\cdot]_{E'})$ be two isomorphic Courant algebroids, that is $\F \in \Hom(E,E')$ satisfies (\ref{eq_iso}). Let $\cD$ be a Courant algebroid connection on $E$, and let $\cD'$ be defined by
\begin{equation} \label{eq_isoconn}
\F( \cD_{\psi}\psi') = \cD'_{\F(\psi)} \F(\psi'), 
\end{equation}
for all $\psi,\psi' \in \Gamma(E)$. Clearly, $\cD'$ is a Courant algebroid connection on $E'$. The assertion of the lemma is that $X_{\cD} = X_{\cD'}$. It follows from (\ref{eq_isoconn}) that $\Div_{\cD}(\psi) = \Div_{\cD'}(\F(\psi))$ for all $\psi \in \Gamma(E)$. Also,  $\D'{f} = \F(\D{f})$ from (\ref{eq_iso}). Consequently, one obtains $\Div_{\cD}(\D{f}) = \Div_{\cD'}(\D'{f})$. 
\end{proof}

There is another tensorial quantity invariant in the same sense as $X_{\cD}$ of the previous lemma. It uses the induced vector bundle map $\rho^{\ast} \in \Hom(T^{\ast}M,E)$. It is a contravariant tensor $V_{\cD} \in \T^{3}_{0}(M)$ defined as 
\begin{equation} \label{eq_VcD}
V_{\cD}(\xi,\eta,\zeta) = \< \cD_{\rho^{\ast}(\xi)} \rho^{\ast}(\eta), \rho^{\ast}(\zeta) \>_{E},
\end{equation}
for all $\xi,\eta,\zeta \in \df{1}$. This quantity is indeed $\cif$-linear in all inputs. For the second one,  $\rho \circ \rho^{\ast} = 0$ must be used. It follows from (\ref{eq_conncomp}) that in fact one has $V_{\cD} \in \df{1} \otimes \df{2}$, that is $V_{\cD}(\xi,\eta,\zeta) + V_{\cD}(\xi,\zeta,\eta) = 0$. Moreover, $V_{\cD}$ is invariant under Courant algebroid isomorphisms, i.e.,  $V_{\cD}= V_{\cD'}$.

\begin{rem}
It might seem that the bracket $[\cdot,\cdot]_{E}$ is irrelevant for the definitions of $X_{\cD}$ and $V_{\cD}$. This is not true, as equations in (\ref{eq_courantprops}) were proved using the left Leibniz rule (\ref{eq_leftleibniz}) combined with (\ref{eq_rhoishom}). However, everything still works for pre-Courant algebroids mentioned in Remark \ref{rem_precourant}. 
\end{rem}

The Main idea behind Courant algebroid connections generalizing the ordinary vector bundle connections is to have an $\R$-bilinear map mapping a pair of sections of $E$ to a section of $E$. It thus resembles standard manifold connections. This suggests an attempt to define an analogue of the torsion operator. However, the naive guess $T(\psi,\psi') = \cD_{\psi}\psi' - \cD_{\psi'}\psi - [\psi,\psi']_{E}$ fails. The map $T$ is neither skew-symmetric nor $\cif$-linear in $\psi$. A solution to this problem was proposed independently in \cite{alekseevxu} and \cite{2007arXiv0710.2719G}. Both definitions are equivalent for Courant algebroid connections, i.e., those satisfying (\ref{eq_conncomp}). A downside of this notion of torsion is its unclear geometrical interpretation.

\begin{definice}
Let $(E,\rho,\<\cdot,\cdot\>_{E},[\cdot,\cdot]_{E})$ be a Courant algebroid, and let $\cD$ be a Courant algebroid connection. Then a \textbf{torsion operator} is a map $T: \Gamma(E) \times \Gamma(E) \rightarrow \Gamma(E)$ defined as 
\begin{equation} \label{eq_torsionop}
T(\psi,\psi') = \cD_{\psi}\psi' - \cD_{\psi'}\psi - [\psi,\psi']_{E} + \< \cD_{\psi_{\lambda}} \psi, \psi'\>_{E} \cdot \psi^{\lambda}_{E},
\end{equation}
for all $\psi,\psi' \in \Gamma(E)$, where $\{\psi_{\lambda}\}_{\lambda=1}^{\rank(E)}$ is an arbitrary local frame on $E$, $\{\psi^{\lambda}\}_{\lambda=1}^{\rank(E)}$ the corresponding dual one. We denote by $\psi^{\lambda}_{E}$ the inverse image of the dual frame under the isomorphism $g_{E}$, that is $\psi^{\lambda}_{E} := g_{E}^{-1}(\psi^{\lambda}),$ for each $\lambda$. 
\end{definice}

\begin{lemma}
The torsion operator $T$ defined by (\ref{eq_torsionop}) is skew-symmetric and $\cif$-linear in both inputs. It thus defines a \textbf{torsion tensor} $T \in \T_{2}^{1}(E)$. Moreover, one can define the covariant tensor $T_{G} \in \Omega^{2}(E) \otimes \Omega^{1}(E)$ as 
\begin{equation}
T_{G}(\psi,\psi',\psi'') = \< T(\psi,\psi'), \psi''\>_{E} \equiv \< \cD_{\psi}\psi' - \cD_{\psi'}\psi - [\psi,\psi']_{E},\psi''\>_{E} + \< \cD_{\psi''}\psi,\psi'\>_{E},
\end{equation}
for all $\psi,\psi',\psi'' \in \Gamma(E)$. Such a $T_{G}$ is completely skew-symmetric, and one calls $T_{G} \in \Omega^{3}(E)$ the (Gualtieri) \textbf{torsion $3$-form}. 
\end{lemma}
\begin{proof}
This is a straightforward calculation using the definitions of the Courant algebroid connection (\ref{eq_connleib}, \ref{eq_conncomp}) and of the Courant algebroid bracket (\ref{eq_lebnizrule}, \ref{eq_courantsympart}, \ref{eq_leftleibniz}). For example, the skew-symmetry of the operator $T$ in its inputs can be proved as follows:
\begin{equation}
T_{G}(\psi,\psi,\psi') = \< \cD_{\psi'}\psi,\psi\>_{E} -\< [\psi,\psi]_{E}, \psi'\>_{E} = \frac{1}{2} \rho(\psi').\<\psi,\psi\>_{E} - \<[\psi,\psi]_{E},\psi'\>_{E} = 0,
\end{equation}
for all $\psi,\psi' \in \Gamma(E)$, where the last equality follows from (\ref{eq_couraxjinak}). Hence $T(\psi,\psi) = 0$, which proves the skew-symmetry. The proof of other assertions is similar. Note that only Leibniz rule and axiom (\ref{eq_courantsympart}) are used in the proof, one does not need neither the Leibniz identity nor the property (\ref{eq_rhoishom}). 
\end{proof}

With the torsion operator issue successfully resolved, we can turn our attention to the curvature operator. Again, the naive definition brings up essentially the same issues as in the case of the torsion operator. Define an operator $R^{(0)}$ as for an ordinary linear connection on a manifold:
\begin{equation}
R^{(0)}(\psi,\psi') \phi = \cD_{\psi}\cD_{\psi'} \phi - \cD_{\psi'}\cD_{\psi} \phi - \cD_{[\psi,\psi']_{E}} \phi,
\end{equation}
for all $\psi,\psi',\phi,\phi' \in \Gamma(E)$. This operator is not skew-symmetric in $(\psi,\psi')$, and moreover the $\cif$-linearity  in the first input $\psi$ is broken due to the more complicated left Leibniz rule (\ref{eq_leftleibniz}). To be more precise, one obtains 
\begin{equation} \label{eq_R0flinearity}
R^{(0)}(f \psi,\psi') \phi = f R^{(0)}(\psi,\psi') \phi - \<\psi,\psi'\>_{E} \cdot \cD_{\D{f}}( \phi). 
\end{equation}
One can circumvent this inconvenience using various approaches. As an example, one can consider $\psi,\psi' \in \Gamma(L)$, where $L \subseteq E$ is a Dirac structure in $E$, i.e., a subbundle maximally isotropic with respect to $\<\cdot,\cdot\>_{E}$ and involutive with respect to $[\cdot,\cdot]_{E}$. Note that $L$ with restricted anchor and bracket is always a Lie algebroid, and $R^{(0)}$ restricted on $\psi,\psi' \in \Gamma(L)$ is then the curvature of the Lie algebroid connection $\cD|_{L}$ on $E$. If $R^{(0)}$ vanishes, such a flat connection defines a Lie algebroid action of $L$ on the vector bundle $E$.  Alternatively, one can take $\psi$ and $\psi'$ to be sections of mutually orthogonal subbundles, such as e.g. $V_{\pm}$ defined by generalized Riemannian metric, see Definition \ref{def_genmet2}. This is an approach pursued e.g. in \cite{2013arXiv1304.4294G}. 

One can follow a completely different path in which one modifies $R^{(0)}$ in order to make it into a honest tensor on $E$. This definition does not require any additional structure.  A convenient definition can be found in  the work of Hohm and Zwiebach \cite{Hohm:2012mf} on double field theory. Up to prefactors, the following definition follows their idea. It is convenient to work with covariant tensors instead of operators. To be more specific, define $R^{(0)}: \Gamma(E)^{\otimes 4} \rightarrow \cif$ as 
\begin{equation}
R^{(0)}(\phi',\phi,\psi,\psi'):= \< R^{(0)}(\psi,\psi')\phi, \phi' \>_{E},
\end{equation}
for all $\psi,\psi',\phi,\phi' \in \Gamma(E)$. We use the same letter for both objects. We hope that this will cause no confusion. 
\begin{definice}
Let $(E,\rho,\<\cdot,\cdot\>_{E},[\cdot,\cdot]_{E})$ be a Courant algebroid, and let $\cD$ be a Courant algebroid connection. Then the \textbf{Riemann curvature tensor} $R \in \T_{4}^{0}(E)$ of the connection $\cD$ is defined as
\begin{equation} \label{eq_Rtensor}
R(\phi',\phi,\psi,\psi') = \frac{1}{2} \{ R^{(0)}(\phi',\phi,\psi,\psi') + R^{(0)}(\psi',\psi,\phi,\phi') + \< \cD_{\psi_{\lambda}} \psi,\psi'\>_{E} \cdot \< \cD_{\psi^{\lambda}_{E}} \phi, \phi'\>_{E} \},
\end{equation}
for all $\psi,\psi',\phi,\phi' \in \Gamma(E)$. The \textbf{Riemann curvature operator} is the operator $R$ related to the Riemann curvature tensor as $\<R(\psi,\psi')\phi, \phi'\>_{E} = R(\phi',\phi,\psi,\psi')$ and denoted by the same symbol. 
\end{definice}
\begin{rem}
In the definitions presented here, we choose to not use the adjective "generalized" in order to declutter the written text. It should be always clear from the context which kind of objects we have in mind. Moreover, we sometimes omit the words Riemann\footnote{Sorry, Bernhard.} or curvature (but never both). 
\end{rem}
The local frame in the above definition $\{\psi_{\lambda}\}_{\lambda=1}^{\rank(E)}$ on $E$ is again arbitrary, and $\{ \psi^{\lambda}_{E} \}_{\lambda=1}^{\rank(E)}$ is the induced local frame defined uniquely by relations $\< \psi_{\lambda}, \psi^{\mu}_{E} \>_{E} = \delta_{\lambda}^{\mu}$. Apart from the $\cif$-linearity in all inputs,  this definition gives a tensor $R$ with interesting symmetries, similar\footnote{But not the same!} to those of ordinary Riemann curvature tensor. We summarize these observations in the form of a proposition.

\begin{tvrz}
The map (\ref{eq_Rtensor}) is $\cif$-linear in all inputs, hence indeed $R \in \T_{4}^{0}(E)$. Moreover, it possesses the following symmetries:
\begin{align}
\label{eq_Rsym1} R(\phi',\phi,\psi,\psi') + R(\phi',\phi,\psi',\psi) & = 0, \\
\label{eq_Rsym2} R(\phi',\phi,\psi,\psi') + R(\phi,\phi',\psi,\psi') & = 0, \\
\label{eq_Rsym3} R(\phi',\phi,\psi,\psi') - R(\psi',\psi,\phi,\phi') & = 0, 
\end{align}
for all $\psi,\psi',\phi,\phi' \in \Gamma(E)$. In particular, the curvature operator $R(\psi,\psi')$ is skew-symmetric in $(\psi,\psi')$. Moreover, (\ref{eq_Rsym1} - \ref{eq_Rsym3}) imply the interchange symmetry:
\begin{equation}
\label{eq_Rsym4} R(\phi',\phi,\psi,\psi') - R(\psi,\psi',\phi',\phi) = 0. 
\end{equation}
\end{tvrz}
\begin{proof}
From (\ref{eq_lebnizrule}, \ref{eq_rhoishom}, \ref{eq_leftleibniz}) and (\ref{eq_connleib}) it follows that $R^{(0)}$ is $\cif$-linear in all inputs except for the third one, where (\ref{eq_R0flinearity}) applies:
\begin{equation}
R^{(0)}(\phi',\phi,f \psi, \psi') = f R^{(0)}(\phi',\phi, \psi, \psi') - \< \psi,\psi'\>_{E} \cdot \<\cD_{\D{f}}(\phi), \phi'\>_{E}.
\end{equation}
But this is exactly corrected by the third term in (\ref{eq_Rtensor}):
\begin{equation}
\<\cD_{\psi_{\lambda}}(f\psi), \psi'\>_{E} \cdot \< \cD_{\psi^{\lambda}_{E}}\phi,\phi'\>_{E} = \{ f \<\cD_{\psi_{\lambda}}(\psi), \psi'\>_{E} + (\rho(\psi_{\lambda}).f) \cdot \<\psi,\psi'\>_{E} \} \cdot \< \cD_{\psi^{\lambda}_{E}}\phi,\phi'\>_{E} \}.
\end{equation}
Note that $\rho(\psi_{\lambda}).f = \<\psi_{\lambda}, \D{f}\>_{E}$ and $\<\psi_{\lambda},\D{f}\>_{E} \cdot \psi^{\lambda}_{E} = \D{f}$. The second term of the expression on the right-hand side thus gives precisely $\<\psi,\psi'\>_{E} \cdot \<\cD_{\D{f}}(\phi), \phi'\>_{E}$. Similarly, the second copy of $R^{(0)}$ in (\ref{eq_Rtensor}) is $\cif$-linear in all inputs except for $\phi$, which is again corrected by the third term. This term itself is $\cif$-linear in the remaining two inputs $\psi'$ and $\phi'$. Hence $R \in \T_{4}^{0}(E)$. 

Next, the symmetries. The one in (\ref{eq_Rsym3}) is manifest, following directly from the definition (\ref{eq_Rtensor}). To prove (\ref{eq_Rsym1}), one first shows that the map $R^{(0)}$ in fact satisfies (\ref{eq_Rsym2}). This is follows by repeated use of (\ref{eq_conncomp}) together with (\ref{eq_rhoishom}). With the help of this observation, one has
\begin{align}
R^{(0)}(\phi',\phi, \psi,\psi) & = -\< \cD_{[\psi,\psi]_{E}} \phi, \phi'\>_{E}, \\
R^{(0)}(\psi, \psi, \phi , \phi') & = 0,
\end{align}
for all $\psi, \phi, \phi' \in \Gamma(E)$. For the third term, one finds
\begin{equation}
\begin{split}
\< \cD_{\psi_{\lambda}} \psi,\psi\>_{E} \cdot \<\cD_{\psi^{\lambda}_{E}} \phi, \phi'\>_{E} = & \ \frac{1}{2} (\rho(\psi_{\lambda}).\<\psi,\psi\>_{E}) \cdot \< \cD_{\psi^{\lambda}_{E}} \phi, \phi'\>_{E} \\
= & \ \frac{1}{2} \< \D\<\psi,\psi\>_{E}, \psi_{\lambda}\>_{E} \cdot \<\cD_{\psi^{\lambda}_{E}}\phi,\phi'\>_{E} \\
= & \ \< \cD_{ \frac{1}{2} \D{\<\psi,\psi\>_{E}} } \phi, \phi' \>_{E}. 
\end{split}
\end{equation}
Here we have used the metric compatibility (\ref{eq_conncomp}). Summing up all three expressions, one obtains
\begin{equation}
R(\phi',\phi,\psi,\psi) = \frac{1}{2}\< \cD_{ [\psi,\psi]_{E} - \frac{1}{2} \D{\<\psi,\psi\>_{E}}} \phi, \phi'\>_{E} = 0,
\end{equation}
as there holds the Courant algebroid axiom in the form (\ref{eq_couraxjinak}). This proves (\ref{eq_Rsym1}). The symmetry in (\ref{eq_Rsym2}) now in fact follows from (\ref{eq_Rsym1}) combined with (\ref{eq_Rsym3}). The rest of the assertions easily follows  and the proposition is now proved. 
\end{proof}
\begin{rem}
Note that opposed to the definition of $T$, one needs (\ref{eq_rhoishom}) to hold in order to obtain a tensorial quantity $R$. The metric compatibility (\ref{eq_conncomp}) is only required in order to have the symmetries (\ref{eq_Rsym1}, \ref{eq_Rsym2}) and their consequence (\ref{eq_Rsym4}). Note that the symmetry (\ref{eq_Rsym3}) is new compared to the Riemannian geometry, and it implies (\ref{eq_Rsym4}) which for ordinary Riemann tensor holds only if $\cD$ is torsion-free. 
\end{rem}

Symmetries of the Riemann tensor are important for an unambiguous definition of the Ricci tensor, the contraction of $R$ in two indices. Moreover, note that one always has  a fiber-wise metric $\<\cdot,\cdot\>_{E}$ at disposal to raise  indices. 
\begin{definice}
Let $(E,\rho,\<\cdot,\cdot\>_{E},[\cdot,\cdot]_{E})$ be a Courant algebroid, and let $\cD$ be a Courant algebroid connection. Then the \textbf{Ricci curvature tensor} $\Ric \in \T_{2}^{0}(E)$ is defined as 
\begin{equation}
\Ric(\psi,\psi') = \< \psi^{\lambda}, R(\psi_{\lambda},\psi')\psi \>_{E} \equiv R(\psi^{\lambda}_{E}, \psi, \psi_{\lambda}, \psi'),
\end{equation}
for all $\psi,\psi' \in \Gamma(E)$. It is symmetric in $(\psi,\psi')$ and all other contractions in two indices of $R$ are either zero or proportional to $\Ric$. Moreover, one defines the \textbf{Courant-Ricci scalar} $\RS_{E}$ as
\begin{equation}
\RS_{E} = \Ric(\psi^{\lambda}_{E}, \psi_{\lambda}) \equiv R(\psi^{\mu}_{E}, \psi^{\lambda}_{E}, \psi_{\mu}, \psi_{\lambda}).
\end{equation}
We use the name Courant-Ricci as to indicate that $\<\cdot,\cdot\>_{E}$ is used for calculation of the trace. 
\end{definice}
Although, at the moment, the definitions of $R$ and $T$ lack a clear geometric interpretation, they interplay together in the following analogue of the algebraic Bianchi identity. This was proved  in a similar fashion in \cite{Hohm:2012mf}. Notice that the right-hand side is different from the ordinary manifold case. 

\begin{theorem}[Algebraic Bianchi identity] Let $\cD$ be a Courant algebroid connection. Then its curvature operator $R$ satisfies the identity
\begin{equation} \label{eq_bianchi} \begin{split}
R(\psi,\psi')\psi'' + cyclic(\psi,\psi',\psi'') = & \ \frac{1}{2} \{ (\cD_{\psi}T_{G})(\psi',\psi'',\psi_{\lambda}) \cdot \psi^{\lambda}_{E} - T(\psi,T(\psi',\psi'')) \\
& + cyclic(\psi,\psi',\psi'') - (\cD_{\psi_{\lambda}} T_{G})(\psi,\psi',\psi'') \cdot \psi^{\lambda}_{E} \},
\end{split}
\end{equation}
for all $\psi,\psi',\psi'' \in \Gamma(E)$. In particular, if $T = 0$, one has 
\begin{equation}
R(\psi,\psi')\psi'' + cyclic(\psi,\psi',\psi'') = 0.
\end{equation}
\end{theorem}
\begin{proof}
The proof is rather technical but straightforward. In the proof, we will use the abbreviation $cyc$ for $cyclic(\psi,\psi',\psi'')$. 
Let us start with reordering the Leibniz identity (\ref{eq_leibnizidentity}). One uses (\ref{eq_courantsympart}) and (\ref{eq_courantprops}) to find 
\begin{equation} 
[\psi,[\psi',\psi'']_{E}]_{E} + cyc = \D\< [\psi,\psi']_{E},\psi''\>_{E} + \D\<\psi', \D{\<\psi,\psi''\>_{E}}\>_{E}.
\end{equation}
Contracting this with $\phi \in \Gamma(E)$ gives 
\begin{equation} \label{eq_leibnizreordered}
\< [\psi,[\psi',\psi'']_{E}]_{E}, \phi\>_{E} + cyc = \rho(\phi).\<[\psi,\psi']_{E}, \psi''\>_{E} + \rho(\phi).\<\psi', \D{\<\psi,\psi''\>_{E}} \>_{E}.
\end{equation}
We will now derive the Bianchi identity for the map $R^{(0)}$. This is analogous to the usual proof, except that we now have to use $\cD_{\psi}\psi' - \cD_{\psi'}\psi = [\psi,\psi']_{E} + T(\psi,\psi') - \fK(\psi,\psi')$, where $\fK(\psi,\psi') = \< \cD_{\psi_{\lambda}}\psi,\psi'\>_{E} \cdot \psi^{\lambda}_{E}$. One obtains
\begin{equation}
\begin{split}
R^{(0)}(\phi,\psi'',\psi,\psi') + cyc = & \ \< \cD_{\psi}( \cD_{\psi'}\psi'' - \cD_{\psi''}\psi'), \phi\>_{E} - \< \cD_{[\psi,\psi']_{E}}\psi'', \phi\>_{E} + cyc \\
= & \ \< \cD_{\psi}( [\psi',\psi'']_{E} + T(\psi',\psi'') - \fK(\psi',\psi'')), \phi\>_{E} + cyc \\
  & - \< \cD_{[\psi,\psi']_{E}}, \psi'', \phi\>_{E} + cyc \\
= & \ \< \cD_{\psi}(T(\psi',\psi'')) + T(\psi,[\psi',\psi'']_{E}),\phi\>_{E} + cyc \\
& - \< \cD_{\psi} (\fK(\psi',\psi'')), \phi\>_{E} - \< \fK(\psi,[\psi',\psi'']_{E}), \phi\>_{E} + cyc \\
& + \<[\psi,[\psi',\psi'']_{E}]_{E}, \phi\>_{E} + cyc. 
\end{split}
\end{equation}
Using (\ref{eq_leibnizreordered}) and the definition of $\fK$, we can rewrite the above result as
\begin{equation}
\begin{split}
R^{(0)}(\phi,\psi'',\psi,\psi') + cyc = & \ \< \fT(\psi,\psi',\psi''), \phi \>_{E} \\
& - \< \cD_{\psi}( \fK(\psi',\psi'')), \phi \>_{E} - \< \cD_{\phi}\psi, [\psi',\psi'']_{E} \>_{E} + cyc \\
& + \rho(\phi).\<[\psi,\psi']_{E},\psi''\>_{E} + \rho(\phi).\<\psi',\D\<\psi,\psi''\>_{E}\>_{E},
\end{split}
\end{equation}
where by $\fT$ we have denoted the following expression 
\begin{equation}
\fT(\psi,\psi',\psi'') = \< \cD_{\psi}( T(\psi',\psi'')) + T(\psi,[\psi',\psi'']_{E}), \phi\>_{E} + cyc. 
\end{equation}
For the ordinary manifold case, this is exactly the right-hand side of the Bianchi identity. However, for a Courant algebroid connection, $\fT$ is not a tensor. Now, note that
\begin{equation}
\< \cD_{\psi_{\lambda}}\psi,\psi'\>_{E} \cdot \< \cD_{\psi^{\lambda}_{E}}\psi'',\phi\>_{E} + cyc = \< \cD_{\fK(\psi',\psi'')}\psi, \phi\>_{E} + cyc,
\end{equation}
and we will thus combine this term with the Bianchi identity for $R^{(0)}$ above in order to find
\[
\begin{split}
R^{(0)}(\phi,\psi'',\psi,\psi') + \<\cD_{\fK(\psi',\psi'')}\psi,\phi\>_{E} + cyc = & \ \< \fT'(\psi,\psi',\psi''), \phi \>_{E} \\
& +  \ \< [\fK(\psi,\psi'), \psi'']_{E}, \phi \>_{E} + cyc \\
& - \< \cD_{\phi}\psi, [\psi',\psi'']_{E}\>_{E} + cyc \\
& - \< \cD_{\phi}( \fK(\psi',\psi'')), \psi \>_{E} + cyc \\
& + \rho(\phi).\<[\psi,\psi']_{E},\psi''\>_{E} + \rho(\phi).\<\psi',\D\<\psi,\psi''\>_{E}\>_{E},
\end{split}
\]
where $\fT'$ is a covariant tensor on $E$ defined by 
\begin{equation}
\fT'(\psi,\psi',\psi'') = \< \cD_{\psi}(T(\psi',\psi'')) + T(\psi, [\psi',\psi'']_{E} - \fK(\psi',\psi'')), \phi\>_{E} + cyc. 
\end{equation}
Now, one has to deal with the remaining terms. Recall that $\< \fK(\psi,\psi'), \phi \>_{E} = \<\cD_{\phi}\psi, \psi'\>_{E}$. The most complicated is the second one. Using (\ref{eq_invarianceznovu}), (\ref{eq_courantsympart}) and the metric compatibility (\ref{eq_conncomp}), one finds 
\[
\begin{split}
\< \fK(\psi,\psi'), \psi'']_{E}, \phi\>_{E} = & \ \rho(\phi).\< \fK(\psi,\psi'), \psi''\>_{E} - \< [\psi'', \fK(\psi,\psi')]_{E}, \phi \>_{E} \\
= & \ \rho(\phi).\< \fK(\psi,\psi'), \psi''\>_{E} - \rho(\psi''). \< \fK(\psi,\psi'), \phi \>_{E} + \< \fK(\psi,\psi'), [\psi'',\phi]_{E} \>_{E} \\
= & \ -R^{(0)}(\psi',\psi,\psi'',\phi) + \< \cD_{\psi''}\psi, \cD_{\phi}\psi'\>_{E} - \<\cD_{\phi}\psi, \cD_{\psi''}\psi'\>_{E}.
\end{split}
\]
Summing both sides over the cyclic permutations, one finds 
\begin{equation}
\begin{split}
\< [\fK(\psi,\psi'),\psi'']_{E},\phi\>_{E} + cyc = & \ -R^{(0)}(\psi',\psi,\psi'',\phi) + \< \cD_{\psi}\psi' - \cD_{\psi'}\psi, \cD_{\phi}\psi'' \>_{E} + cyc \\
= & \ - R^{(0)}(\psi',\psi,\psi'',\phi) + \< T(\psi,\psi') + [\psi,\psi']_{E}, \cD_{\phi}\psi''\>_{E} + cyc \\
& - \< \fK(\psi,\psi'), \cD_{\phi}\psi''\>_{E} + cyc. 
\end{split}
\end{equation}
This means that this term cancels the second copy of $R^{(0)}$. We can now plug into the definition:
\begin{equation}
\begin{split}
R(\phi,\psi'',\psi,\psi') + cyc = & \ \frac{1}{2} \big\{ \< \fT'(\psi,\psi',\psi''), \phi \>_{E} \\
& - \< \cD_{\phi}( \fK(\psi,\psi')), \psi''\>_{E} - \< \fK(\psi,\psi'), \cD_{\phi}\psi''\>_{E} + cyc \\
& + \rho(\phi).\<[\psi,\psi']_{E},\psi''\>_{E} + \rho(\phi).\<\psi', \D{\<\psi,\psi''\>_{E}}\>_{E} \\
& + \< T(\psi,\psi'), \cD_{\phi}\psi''\>_{E} + cyc \big\}. 
\end{split}
\end{equation}
The terms on the second line can be rewritten using the metric compatibility (\ref{eq_conncomp}) and combined with the two terms on the third line to find 
\begin{equation}
- \rho(\phi). \<\cD_{\psi}\psi' - \cD_{\psi'}\psi - [\psi,\psi']_{E} + \fK(\psi,\psi'), \psi''\>_{E} = - \rho(\phi).T_{G}(\psi,\psi',\psi''). 
\end{equation}
Finally, this combines with the very last line to give the tensorial expression
\begin{equation}
R(\phi,\psi'',\psi,\psi') + cyc = \frac{1}{2} \{ \< \fT'(\psi,\psi',\psi''), \phi\>_{E} - (\cD_{\phi}T_{G})(\psi,\psi',\psi'') \}. 
\end{equation}
To finish the proof, it remains to prove that the first term $\fT'$ can be rewritten to match (\ref{eq_bianchi}). This is straightforward, as $[\psi,\psi']_{E} - \fK(\psi,\psi') = - T(\psi,\psi') + \cD_{\psi}\psi' - \cD_{\psi'}\psi$, and thus 
\begin{equation}
\begin{split}
\< \fT'(\psi,\psi',\psi''),\phi\>_{E} = & \ \< \cD_{\psi}(T(\psi',\psi'')), \phi \>_{E} + T_{G}(\psi, \cD_{\psi'}\psi'' - \cD_{\psi''}\psi' - T(\psi',\psi''), \phi) + cyc \\
= & \ (\cD_{\psi}T_{G})(\psi',\psi'',\phi) - T_{G}(\psi,T(\psi',\psi''),\phi) + cyc.
\end{split}
\end{equation}
Plugging into the above formula, we conclude that 
\begin{equation}
\begin{split}
R(\phi,\psi'',\psi,\psi') + cyclic(\psi,\psi',\psi'') = & \ \frac{1}{2} \{ (\cD_{\psi}T_{G})(\psi',\psi'',\phi) - T_{G}(\psi,T(\psi',\psi''),\phi) \\
& + cyclic(\psi,\psi',\psi'') - (\cD_{\phi}T_{G})(\psi,\psi',\psi'') \},
\end{split}
\end{equation}
which is exactly the algebraic Bianchi identity (\ref{eq_bianchi}). 
\end{proof}
\begin{rem}
We have included this painful proof mainly in order to demonstrate how non-trivially the axioms of Courant algebroid and of the metric compatibility (\ref{eq_conncomp}) interplay to give the result fully expressible in terms of the torsion $3$-form $T_{G}$ and its covariant derivatives. In particular, note that unlike anywhere before, one uses the Leibniz identity (\ref{eq_leibnizidentity}) in the proof. 
\end{rem}
To conclude this section, we point out one very important property of the tensors $T$ and $R$, namely how they transform under Courant algebroid isomorphisms. The answer supports our arguments why it is a good idea to take the definitions (\ref{eq_torsionop}) and (\ref{eq_Rtensor}) seriously. 
\begin{tvrz} \label{tvrz_RGcovariance}
Let $(E,\rho,\<\cdot,\cdot\>_{E},[\cdot,\cdot]_{E})$ and $(E',\rho',\<\cdot,\cdot\>_{E'}, [\cdot,\cdot]_{E'})$ be two isomorphic Courant algebroids, and let $\F \in \Hom(E,E')$ be the isomorphism. Assume that $\cD$ and $\cD'$ are two connections related as in (\ref{eq_isoconn}). Let $(T_{G},R)$ correspond to $\cD$ and let $(T'_{G},R')$ correspond to $\cD'$. 

Then $T_{G} = \F^{\ast} T'_{G}$ and $R = \F^{\ast} R'$. Consequently, $\Ric = \F^{\ast} \Ric'$ and $\RS_{E} = \RS_{E'}$. 
\end{tvrz}
\begin{proof}
This is a straightforward calculation using (\ref{eq_iso}) and (\ref{eq_isoconn}). For example, one has 
\begin{equation}
\begin{split}
T_{G}(\psi,\psi',\psi'') = & \ \< \cD_{\psi}\psi' - \cD_{\psi'}\psi - [\psi,\psi']_{E}, \psi''\>_{E} + \< \cD_{\psi''}\psi,\psi'\>_{E} \\
= & \ \< \F\{ \cD_{\psi}\psi' - \cD_{\psi'}\psi -[\psi,\psi']_{E} \}, \F(\psi'')\>_{E'} + \< \F(\cD_{\psi''}\psi), \F(\psi')\>_{E'} \\
= & \ \< \cD'_{\F(\psi)} \F(\psi') - \cD'_{\F(\psi')} \F(\psi) - [\F(\psi),\F(\psi')]_{E'}, \F(\psi'') \>_{E'} \\
& + \< \cD'_{\F(\psi'')} \F(\psi), \F(\psi') \>_{E'} \\
= & \ T'_{G}( \F(\psi),\F(\psi'),\F(\psi'')),
\end{split}
\end{equation}
for all $\psi,\psi',\psi'' \in \Gamma(E)$. The proof for $R$ is similar.
\end{proof}
\section{Levi-Civita connections} \label{sec_LCconnections}
Assume that $(E,\rho,\<\cdot,\cdot\>_{E},[\cdot,\cdot]_{E})$ is a Courant algebroid, equipped with a generalized Riemannian metric $\gm$. It is thus natural to consider Courant algebroid connections which are compatible with the fiber-wise metric $\gm$. First, we can reinterpret this requirement in terms of other structures induced by a generalized metric. 
\begin{lemma}
Let $\cD$ be a Courant algebroid connection. 

Then the following statements are equivalent:
\begin{enumerate}
\item $\cD$ is compatible with $\gm$, that is $\cD \gm = 0$. 
\item $\cD$ commutes with the map $\tau$, that is $\cD_{\psi}(\tau(\psi')) = \tau( \cD_{\psi}\psi')$ for all $\psi,\psi' \in \Gamma(E)$.
\item $\cD$ preserves the subbundle $V_{+}$, that is $\cD_{\psi}(V_{+}) \subseteq V_{+}$ for all $\psi \in \Gamma(E)$.
\end{enumerate}
\end{lemma}
\begin{proof}
First, prove $1. \Rightarrow 2.$, so assume that $\cD \gm = 0$. This and (\ref{eq_conncomp}) imply that for every $\psi \in \Gamma(E)$,  $\cD_{\psi}$ commutes with both isomorphisms $g_{E} \in \Hom(E,E^{\ast})$ and $\gm \in \Hom(E,E^{\ast})$, induced by the fiber-wise metrics. As $\tau = g_{E}^{-1} \gm$, this proves the assertion. The implication $2. \Rightarrow 3.$ is trivial, as $V_{+}$ is the $+1$ eigenbundle of $\tau$. To show $3. \Rightarrow 1.$, one first proves that $\cD_{\psi}$ preserves also $V_{-}$. This follows from (\ref{eq_conncomp}) and the fact that $V_{-} = V_{+}^{\perp}$. 
This proves that for each $\psi \in \Gamma(E)$, $\cD_{\psi}$ is block-diagonal with respect to the decomposition $E = V_{+} \oplus V_{-}$. The same holds for $g_{E}$ and $\gm$, namely, we have the following formal block forms of the involved objects:
\begin{equation} \label{eq_blocksconnmetrics}
\cD_{\psi} = \bm{\cD_{\psi}^{+}}{0}{0}{\cD_{\psi}^{-}}, \; \; g_{E} = \bm{g_{E}^{+}}{0}{0}{g_{E}^{-}}, \; \; \gm = \bm{g_{E}^{+}}{0}{0}{-g^{-}_{E}},
\end{equation}
where by $\cD^{\pm}_{\psi}$ and $g_{E}^{\pm}$ we denote the induced objects on $V_{\pm}$. Note that by construction of generalized metric, $g_{E}^{+} > 0$ and $g_{E}^{-} < 0$. (\ref{eq_conncomp}) is then equivalent to $\cD^{\pm}_{\psi}(g^{\pm}_{E}) = 0$. This in turn implies $\cD_{\psi}\gm = 0$, as the only difference is the sign in front of $g_{E}^{-}$. 
\end{proof}

This proof also answers the question about existence of a Courant algebroid connection compatible with the generalized metric $\gm$. One simply has to find a pair of vector bundle connections $\cD'^{\pm}$ on $V_{\pm}$ compatible with $g_{E}^{\pm}$. This is always possible and we can, using (\ref{eq_blocksconnmetrics}), construct a vector bundle connection $\cD'$ on $E = V_{+} \oplus V_{-}$ compatible with both $g_{E}$ and $\gm$. Then set $\cD_{\psi} = \cD'_{\rho(\psi)}$ to obtain an example of a Courant algebroid connection compatible with $\gm$. 

As we have introduced the operator of torsion, it is natural to consider connections which are torsion-free. Note that, for the consistency of this condition, it is important that $T$ is indeed a tensor. To follow the terminology in Riemannian geometry, we refer to such connections as Levi-Civita connections.
\begin{definice}
Let $(E,\rho,\<\cdot,\cdot\>_{E},[\cdot,\cdot]_{E})$ be a Courant algebroid, let $\gm$ be a generalized metric on $E$, and let $\cD$ be a Courant algebroid connection. 

We say that $\cD$ is a \textbf{Levi-Civita connection} on $E$ with respect to the generalized metric $\gm$, if $\cD \gm = 0$ and $\cD$ is a torsion-free connection, that is $T_{G} = 0$. 
\end{definice}
The question of existence of Levi-Civita connections is more intriguing. We do not have a resolute answer to this question. First, observe that an analogue of the usual formula for the Levi-Civita connection does not even define a connection. One can try to follow the derivation of the ordinary Levi-Civita connection. The crucial point of the derivation is to rewrite the combination $\cD_{\psi}\psi' - \cD_{\psi'}\psi$ using the torsion and the bracket, which eventually leads to the concept of a contortion tensor. This does not work here because the operator $T$ defined by (\ref{eq_torsionop}) is more complicated. As we will now demonstrate, this is rather a conceptual problem - in general, there can be infinitely many Levi-Civita connections. This is shown in the following lemma:
\begin{lem} \label{lem_LCconnections}
Let $(E,\rho,\<\cdot,\cdot\>_{E},[\cdot,\cdot]_{E})$ be a Courant algebroid equipped with a generalized  metric $\gm$. Let $\LC(E,\gm)$ denote the set of all Levi-Civita connections $\cD$ on $E$ with respect to $\gm$. Assume that $\LC(E,\gm) \neq \emptyset$. Then $\LC(E,\gm)$ is an affine space, where its associated vector space is a $\cif$-module of sections $\Gamma(\LC_{0}(E,\gm))$ of a vector bundle $\LC_{0}(E,\gm)$ of rank
\begin{equation} \label{eq_LCEGrank}
\rank(\LC_{0}(E,\gm)) = \frac{1}{3} p(p^{2}-1) + \frac{1}{3}q(q^{2}-1),
\end{equation}
where $(p,q)$ is the signature of the fiber-wise metric $\<\cdot,\cdot\>_{E}$. 
\end{lem}
\begin{proof}
By assumption, there exists at least one $\cD \in \LC(E,\gm)$. Let $\cD' \in \LC(E,\gm)$ be any other Levi-Civita connection. Define a map $\K: \Gamma(E)^{\otimes 3} \rightarrow C^{\infty}(M)$ as 
\begin{equation}
\K(\psi,\psi',\psi'') = \< \cD'_{\psi}\psi' - \cD_{\psi}\psi', \psi''\>_{E}, 
\end{equation}
for all $\psi,\psi',\psi'' \in \Gamma(E)$. The property (\ref{eq_connleib}) of both connections implies that $\K \in \T_{3}^{0}(E)$. Moreover, the compatibility with $\<\cdot,\cdot\>_{E}$ (\ref{eq_conncomp}) of both connections forces $\K \in \Omega^{1}(E) \otimes \Omega^{2}(E)$, that is $\K$ is skew-symmetric in last two inputs. Next, the vanishing torsion $3$-forms for both $\cD$ and $\cD'$ give
\begin{equation} \label{eq_Kcyclic}
\K(\psi,\psi',\psi'') + cyclic(\psi,\psi',\psi'') = 0, 
\end{equation}
for all $\psi,\psi',\psi'' \in \Gamma(E)$. For $K \in \Omega^{1}(E) \otimes \Omega^{2}(E)$, this is in fact equivalent to the requirement $\K_{a} = 0$, where $\K_{a}$ denotes the complete skew-symmetrization of the tensor $\K$. Finally, as $\cD$ and $\cD'$ are both compatible with $\gm$, we find that $\K(\psi,\psi_{+},\psi_{-}) = \K(\psi,\psi_{-},\psi_{+}) = 0$ for all $\psi \in \Gamma(E)$ and $\psi_{\pm} \in \Gamma(V_{\pm})$. Combined with (\ref{eq_Kcyclic}), this shows that $\K$ has to live in the $\cif$-submodule isomorphic to $(\Omega^{1}(V_{+}) \otimes \Omega^{2}(V_{+})) \oplus (\Omega^{1}(V_{-}) \otimes \Omega^{2}(V_{-}))$. In other words, $\K$ has non-trivial values only if all inputs are either from $\Gamma(V_{+})$ or from $\Gamma(V_{-})$. Moreover, we still have to impose (\ref{eq_Kcyclic}). This can be written as 
\begin{equation}
\Gamma(\LC_{0}(E,\gm)) \cong \frac{\Omega^{1}(V_{+}) \otimes \Omega^{2}(V_{+})}{\Omega^{3}(V_{+})} \oplus \frac{\Omega^{1}(V_{-}) \otimes \Omega^{2}(V_{-})}{\Omega^{3}(V_{-})}
\end{equation}
As $\rank(V_{+}) = p$ and $\rank(V_{-}) = q$, the formula (\ref{eq_LCEGrank}) follows. 
\end{proof}
We see that Levi-Civita connections is unique if and only if $p,q \in \{0,1\}$. The proof of the lemma states that given a fixed Levi-Civita connection $\cD \in \LC(E,\gm)$, every other one can be written as 
\begin{equation} \label{eq_LC'asLCandK}
\cD'_{\psi}\psi' = \cD_{\psi}\psi + g_{E}^{-1} \K(\psi,\psi',\cdot),
\end{equation}
for all $\psi,\psi' \in \Gamma(E)$, with $\K$ being a tensor with the above described properties. In the next section, it will be convenient to express the Riemann tensor, as well as the Ricci tensor and the Courant-Ricci scalar of $\cD'$ in terms of $\cD$ and $\K$. Let $\cD'$ and $\cD$ be related by (\ref{eq_LC'asLCandK}). First, by a straightforward calculation, one gets
\begin{equation}
\begin{split}
R'^{(0)}(\phi',\phi,\psi,\psi') = & \ R^{(0)}(\phi',\phi,\psi,\psi') - \K(\fK(\psi,\psi'),\phi,\phi') \\
& + (\cD_{\psi}\K)(\psi',\phi,\phi') - (\cD_{\psi'}\K)(\psi,\phi,\phi') \\
& + \K(\psi, g_{E}^{-1}\K(\psi',\phi,\cdot),\phi') - \K(\psi',g_{E}^{-1}\K(\psi,\phi,\cdot),\phi'). 
\end{split}
\end{equation}
Here we used the torsion-free condition for $\cD$ in order to rewrite the bracket $[\cdot,\cdot]_{E}$ using $\cD$ and the map $\fK(\psi,\psi') = \< \cD_{\psi_{\lambda}}\psi,\psi'\>_{E} \cdot \psi^{\lambda}_{E}$. The relation between the respective third terms in (\ref{eq_Rtensor}) can be recast as
\begin{equation}
\begin{split}
\< \cD'_{\psi_{\lambda}}\psi,\psi'\>_{E} \cdot \< \cD'_{\psi^{\lambda}_{E}}\phi,\phi'\>_{E} = & \ \< \cD_{\psi_{\lambda}}\psi,\psi'\>_{E} \cdot \< \cD_{\psi^{\lambda}_{E}}\phi,\phi'\>_{E} \\
& + \K(\fK(\psi,\psi'),\phi,\phi') + \K(\fK(\phi,\phi'),\psi,\psi') \\
& + \K(g_{E}^{-1}\K(\cdot,\psi,\psi'),\phi,\phi'). 
\end{split}
\end{equation}
The two terms with $\fK$ cancel those coming from the two copies of $R'^{(0)}$ and we find
\begin{equation} \label{eq_R'andR}
\begin{split}
R'(\phi',\phi,\psi,\psi') = & \ R(\phi',\phi,\psi,\psi') + \frac{1}{2}\{(\cD_{\psi}\K)(\psi',\phi,\phi') - (\cD_{\psi'}\K)(\psi,\phi,\phi') \\
& + (\cD_{\phi}\K)(\phi',\psi,\psi') - (\cD_{\phi'}\K)(\phi,\psi,\psi') + \K(g_{E}^{-1}\K(\cdot,\psi,\psi'),\phi,\phi') \\
& + \K(\psi,g_{E}^{-1}\K(\psi',\phi,\cdot),\phi') - \K(\psi',g_{E}^{-1}\K(\psi,\phi,\cdot),\phi') \\
& + \K(\phi,g_{E}^{-1}\K(\phi',\psi,\cdot),\psi') - \K(\phi',g_{E}^{-1}\K(\phi,\psi,\cdot),\psi')\}.
\end{split}
\end{equation}
Observe that this relation is consistent with both sides being tensors on $E$. One can now easily find the relation between the two respective Ricci curvature tensors. Note that one uses the fact that $\cD$ commutes with the contractions using the fiber-wise metric $\<\cdot,\cdot\>_{E}$. The resulting expression is
\begin{equation} \label{eq_Ric'andRic}
\begin{split}
\Ric'(\psi,\psi') = & \Ric(\psi,\psi') + \frac{1}{2}\{(\cD_{\psi_{\lambda}} \K)(\psi,\psi',\psi^{\lambda}_{E}) + (\cD_{\psi_{\lambda}}\K)(\psi',\psi,\psi^{\lambda}_{E}) \\
& + (\cD_{\psi} \K')(\psi') + (\cD_{\psi'} \K')(\psi) - \K'(\psi_{\lambda}) \{ \K(\psi,\psi',\psi^{\lambda}_{E}) + \K(\psi',\psi,\psi^{\lambda}_{E}) \} \\
& + \K(g_{E}^{-1}\K(\cdot,\psi_{\lambda},\psi'),\psi,\psi^{\lambda}_{E}) - \K(\psi',g_{E}^{-1}\K(\psi_{\lambda},\psi,\cdot), \psi^{\lambda}_{E}) \\
& - \K(\psi^{\lambda}_{E}, g_{E}^{-1}\K(\psi,\psi_{\lambda},\cdot), \psi')\}. 
\end{split}
\end{equation}
In the above formula the 1-form $\K'$ on $E$ is defined as $\K':= \K(\psi_{\lambda},\psi^{\lambda}_{E},.)$.
Finally, we can compare the two Courant-Ricci scalars $\RS'_{E}$ and $\RS_{E}$. We introduce the covariant divergence of the $1$-form $\K'$ as $\Div_{\cD}(\K') := (\cD_{\psi_{\lambda}}\K')(\psi^{\lambda}_{E})$ and a (pseudo)norm $\rVert \K' \rVert_{E}$ of the $1$-form $\K'$ with respect to the fiber-wise metric $\<\cdot,\cdot\>_{E}$ as $\rVert \K' \rVert^{2}_{E} = g_{E}^{-1}( \K', \K') = \K'(\psi_{\lambda})K'(\psi^{\lambda}_{E})$. Then
\begin{equation}
\begin{split}
\RS'_{E} = & \ \RS_{E} + 2 \Div_{\cD}(\K') - \rVert \K' \rVert^{2}_{E} \\
& + \frac{1}{2} \{ \K(g_{E}^{-1}\K(\cdot,\psi_{\lambda},\psi_{\mu}), \psi^{\mu}_{E}, \psi^{\lambda}_{E}) - 2 \K(\psi^{\mu}_{E}, g_{E}^{-1}\K(\psi_{\lambda},\psi_{\mu},\cdot), \psi^{\lambda}_{E}) \}. 
\end{split}
\end{equation}
This expression can be significantly simplified using the property (\ref{eq_Kcyclic}) of $\K$. 
\begin{lem}
We have
\begin{equation} \label{eq_Ktraceequations}
\K(g_{E}^{-1}\K(\cdot,\psi_{\lambda},\psi_{\mu}), \psi^{\mu}_{E}, \psi^{\lambda}_{E}) - 2 \K(\psi^{\mu}_{E}, g_{E}^{-1}\K(\psi_{\lambda},\psi_{\mu},\cdot), \psi^{\lambda}_{E}) = 0.
\end{equation}
\end{lem}
\begin{proof}
The equation above can be rewritten as 
\begin{equation}
\K( \psi^{\nu}_{E}, \psi^{\mu}_{E}, \psi^{\lambda}_{E}) \cdot \{ \K(\psi_{\nu},\psi_{\lambda}, \psi_{\mu}) - 2 \K(\psi_{\lambda},\psi_{\nu}, \psi_{\mu}) \} = 0.
\end{equation}
We can now use (\ref{eq_Kcyclic}) to rewrite the terms in the curly brackets to obtain the equation
\begin{equation}
\K( \psi^{\nu}_{E}, \psi^{\mu}_{E}, \psi^{\lambda}_{E}) \cdot \{ \K(\psi_{\mu},\psi_{\nu},\psi_{\lambda}) + \K(\psi_{\lambda},\psi_{\nu},\psi_{\mu}) \} = 0.
\end{equation}
The term in the curly brackets is symmetric in $(\mu,\lambda)$, whereas the term with the upper indices is skew-symmetric in $(\mu,\lambda)$. This proves the equation (\ref{eq_Ktraceequations}).
\end{proof}
We thus find a very simple relation of the two scalar curvatures, which we record in terms of a proposition.
\begin{tvrz} \label{tvrz_RS'asRSandK'}
Let $\cD,\cD' \in \LC(E,\gm)$ be two Levi-Civita connections related by (\ref{eq_LC'asLCandK}). Then the relation of the respective Courant-Ricci scalars has the form
\begin{equation} \label{eq_RS'asRSandK'}
\RS'_{E} = \RS_{E} + 2 \Div_{\cD}(\K') - \rVert \K' \rVert^{2}_{E}.
\end{equation}
\end{tvrz}
As we will show in the following sections, the crucial role will be played by another scalar produced from the Ricci tensor $\Ric$. Indeed, we have also the generalized metric $\gm$ to raise the indices. This leads us to the following definition. 
\begin{definice}
Let $\cD$ be a Courant algebroid connection, and let $\gm$ be a generalized metric on $E$. Then the \textbf{Ricci scalar $\RS_{\gm}$ corresponding to $\gm$} is defined by
\begin{equation} \label{eq_RSG}
\RS_{\gm} = \Ric( \gm^{-1}(\psi^{\lambda}), \psi_{\lambda}). 
\end{equation}
\end{definice}
This curvature is invariant under Courant algebroid isomorphisms orthogonal with respect to generalized metrics. This will be of special importance in the following section. 
\begin{lemma} \label{lem_RGisinvariant}
Let all assumptions of Proposition \ref{tvrz_RGcovariance} hold. Moreover, assume that $\cD \in \LC(E,\gm)$. Let $\gm$ be a generalized metric on $E$ and let $\gm'$ be a generalized metric on $E$, such that (\ref{eq_gmisometry}) holds. 

Then $\cD'$ defined by (\ref{eq_isoconn}) is in $\LC(E',\gm')$, and $\RS'_{\gm'} = \RS_{\gm}$, where $\RS'_{\gm'}$ is the Ricci scalar of $\cD'$ corresponding to $\gm'$. 
\end{lemma}
\begin{proof}
This is a direct consequence of the definition (\ref{eq_RSG}) and of Proposition \ref{tvrz_RGcovariance} together with the isometry condition (\ref{eq_gmisometry}). 
\end{proof}

We are interested in an analogue of (\ref{eq_RS'asRSandK'}) relating the two Ricci scalars $\RS'_{\gm}$ and $\RS_{\gm}$. To achieve this, we will use the decomposition $E = V_{+} \oplus V_{-}$ together with the fact that $\K$ is uniquely determined by its restrictions onto these two subbundles, which restriction we denote as $\K_{\pm}$. In other words, if $I_{\pm}: V_{\pm} \rightarrow E$ are the inclusions of the respective subbundles, one sets $\K_{\pm} = I_{\pm}^{\ast}(\K)$. Using the same maps, one obtains two positive definite fiber-wise metrics $\gm_{\pm}$ on $V_{\pm}$, namely $\gm_{\pm} = I_{\pm}^{\ast}(\gm)$. Using the notation introduced in (\ref{eq_blocksconnmetrics}), this gives $\gm_{+} = g_{E}^{+}$ and $\gm_{-} = - g_{E}^{-}$, where $g_{E}^{\pm} = I_{\pm}^{\ast}(g_{E})$. One can now define the partial traces of $\K_{\pm} \in \Omega^{1}(V_{\pm}) \otimes \Omega^{2}(V_{\pm})$ over $V_{\pm}$ using the respective fiber-wise metrics $\gm_{\pm}$:
\begin{equation}
\K'_{\pm}(\psi) = \K_{\pm}( \psi_{k}, \gm_{\pm}^{-1}( \psi^{k}), \psi), 
\end{equation}
where $\{ \psi_{k} \}_{k=1}^{\rank(V_{\pm})}$ are some local frames on $V_{\pm}$, respectively. Similarly, one introduces the respective covariant divergences and norms:
\begin{equation}
\Div_{\cD^{\pm}}( \K'_{\pm}) = (\cD^{\pm}_{\psi_{k}} \K'_{\pm})(\gm^{-1}_{\pm}(\psi^{k})), \; \; \rVert \K'_{\pm} \rVert_{\gm_{\pm}}^{2} = \K'_{\pm}( \psi_{k}) \cdot \K'_{\pm}( \gm_{\pm}^{-1}(\psi^{k})). 
\end{equation}
Naturally, we could have used similar quantities defined using the metrics $g_{E}^{\pm}$, which would differ only in signs. We can now state the analogue of Proposition \ref{tvrz_RS'asRSandK'} for the scalar curvatures $\RS'_{\gm}$ and $\RS_{\gm}$.
\begin{tvrz} \label{tvrz_RSG'asRSGandK'}
Let $\cD,\cD' \in \LC(E,\gm)$ be two Levi-Civita connections related by (\ref{eq_LC'asLCandK}). Then the relation between the two respective Courant-Ricci scalars has the form
\begin{equation} \label{eq_RSG'asRSGandK'}
\RS'_{\gm} = \RS_{\gm} + 2 \Div_{\cD^{+}}(\K'_{+}) - 2 \Div_{\cD^{-}}(\K'_{-}) - \rVert \K'_{+} \rVert^{2}_{\gm_{+}} - \rVert \K'_{-} \rVert^{2}_{\gm_{-}}.
\end{equation}
\end{tvrz}
\begin{proof}
The proof is analogous to the one leading to (\ref{eq_RS'asRSandK'}). Again, one starts from (\ref{eq_Ric'andRic}), being aware of the fact that $\K$ has no non-trivial components except for those given by $\K_{\pm}$. 
\end{proof}

To conclude this section, we show that there is at least one more interesting property of Levi-Civita connections which is transferred under Courant algebroid isomorphisms preserving the generalized metrics. We will see in the following that it plays a crucial role in physics applications. 
\begin{definice}
Let $\cD$ be a Courant algebroid connection. We say that $\cD$ is \textbf{Ricci compatible with $\gm$}, if its Ricci tensor $\Ric$ is block-diagonal with respect to the decomposition $E = V_{+} \oplus V_{-}$:
\begin{equation}
\Ric(V_{+},V_{-}) = 0. 
\end{equation}
This is equivalent to saying that $\gm^{-1} \Ric \in \End(E)$ preserves either of the subbundles $V_{\pm}$. 
\end{definice}
\begin{lemma} \label{lem_riccompcovariance}
Let all assumptions of Proposition \ref{tvrz_RGcovariance} hold, and let $\gm$ and $\gm'$ be two generalized metrics related as in (\ref{eq_gmisometry}).

Then $\cD$ is Ricci compatible with $\gm$ if and only if $\cD'$ is Ricci compatible with $\gm'$. 
\end{lemma}
\begin{proof}
By definition $\F$ has to preserve the subbundles $V_{\pm}$, that is $\F(V_{\pm}) = V'_{\pm}$. The rest follows from the definition and Proposition \ref{tvrz_RGcovariance}. 
\end{proof}
\begin{rem}
In the following we will impose some conditions on scalar curvatures $\RS_{E}$ and $\RS_{\gm}$ as well as the Ricci compatibility with $\gm$. This does not lead to contradictory requirements, as only the block-diagonal components of $\Ric$ contribute to the Ricci scalars. In particular, imposing the Ricci compatibility gives in general no information about $\RS_{E}$ or $\RS_{\gm}$.
\end{rem}
\begin{example}
Let $(\d,0,\<\cdot,\cdot\>_{\d},[\cdot,\cdot]_{\d})$ be the Courant algebroid corresponding to a quadratic Lie algebra $(\d,\<\cdot,\cdot\>_{\d},[\cdot,\cdot]_{\d})$. Let $\gm$ be a generalized metric defining two linear subspaces $V_{\pm} \subseteq \d$. Let $P_{\pm}: \d \rightarrow V_{\pm}$ be the two projectors and let $x_{\pm} \equiv P_{\pm}(x)$ for all $x \in \d$. Set
\begin{equation}
\begin{split}
\<\cD_{x}y, z\>_{\d} = & \ \frac{1}{3} \<[x_{+},y_{+}]_{\d},z_{+}\>_{\d} + \frac{1}{3} \<[x_{-},y_{-}]_{\d},z_{-}\>_{\d} \\
& + \<[x_{-},y_{+}]_{\d},z_{+}\>_{\d} + \<[x_{+},y_{-}]_{\d},z_{-}\>_{\d},
\end{split}
\end{equation}
for all $x,y,z \in \d$. It is easy to check that $\cD$ is compatible with $\<\cdot,\cdot\>_{\d}$, as the right-hand side is easily seen to be skew-symmetric in $(y,z)$. It also preserves the subspace $V_{+}$, as there is no non-trivial contribution on the right-hand side when $y = y_{+}$ and $z = z_{-}$. It is easy to check that it is also torsion-free. Hence $\cD \in \LC(\gm,\d)$ and this set is non-empty. 
\end{example}
\section{Solution for exact Courant algebroids} \label{sec_exact}
We will now  analyze in detail the set of Levi-Civita connections for $E = \gTM$, equipped with a $H$-twisted Dorfman bracket, see Example \ref{ex_Hdorfman} for the corresponding definitions. Every exact Courant algebroid over a given base manifold $M$ is isomorphic to this case for some closed $H \in \df{3}$, and all involved objects transform expectedly under Courant algebroid isomorphisms. To examine Levi-Civita connections for exact Courant algebroids, it thus suffices to consider this particular case. 

Let $\gm$ be a generalized metric (\ref{eq_gmblock}) corresponding to a pair $(g,B)$. We know that it can be written as a product (\ref{eq_gmasproduct}). Let $\cD$ be a Courant algebroid connection on $(\gTM, \rho, \<\cdot,\cdot\>_{E}, [\cdot,\cdot]_{D}^{H})$. Define the connection $\hcD$ as the twist of $\cD$ by the orthogonal map $e^{B}$ defined by (\ref{eq_etoB}). That is
\begin{equation} \label{eq_conntwist}
e^{B}( \hcD_{\psi}\psi') = \cD_{e^{B}(\psi)} e^{B}(\psi'),
\end{equation}
for all $\psi,\psi' \in \Gamma(E)$. Moreover, note that 
\begin{equation}
e^{B}( [\psi,\psi']_{D}^{H+dB}) = [e^{B}(\psi),e^{B}(\psi')]_{D}^{H}, \; \; \< \psi,\psi'\>_{E} = \<e^{B}(\psi),e^{B}(\psi')\>_{E} 
\end{equation}
hold, for all $\psi,\psi' \in \Gamma(E)$. Also, $\rho(\psi) = \rho(e^{B}(\psi))$, for all $\psi \in \Gamma(E)$. In other words, the map $e^{B}$ is a Courant algebroid isomorphism of $(\gTM,\rho,\<\cdot,\cdot\>_{E},[\cdot,\cdot]_{D}^{H+dB})$ and $(\gTM, \rho, \<\cdot,\cdot\>_{E}, [\cdot,\cdot]_{D}^{H})$. It follows from Proposition \ref{tvrz_RGcovariance} that $\cD \in \LC(\gTM,\gm)$ if and only if $\hcD$ is a Levi-Civita connection on $\gTM$ with $(H+dB)$-twisted Dorfman bracket with respect to the block-diagonal generalized metric $\G$. From now on, let $H' \equiv H + dB$. This shows that it suffices to examine the set $\LC(\gTM,\G)$, where $\gTM$ is endowed with the $H'$-twisted Dorfman bracket. 

In view of Lemma \ref{lem_LCconnections}, the most important question is the existence of at least one Levi-Civita connection. As $\G = \BlockDiag(g,g^{-1})$, it is easy to find a Courant algebroid connection compatible with the generalized metric. Indeed, set 
\begin{equation} \label{eq_hcDL}
\hcDL_{(X,\xi)} = \bm{\cDL_{X}}{0}{0}{\cDL_{X}},
\end{equation}
for all $(X,\xi) \in \Gamma(\gTM)$. Here $\cDL$ is an ordinary Levi-Civita connection on the metric manifold $(M,g)$. We will often use this formal block form for operators on $\Gamma(\gTM)$. For example, the definition (\ref{eq_hcDL}) is equivalent to 
\begin{equation}
\hcDL_{(X,\xi)}(Y,\eta) = (\cDL_{X}Y, \cDL_{X}\eta), 
\end{equation}
for all $(Y,\eta) \in \Gamma(\gTM)$. The connection $\hcDL$ is not torsion-free, as $\widehat{T}^{LC}_{G} = \rho^{\ast}(H')$. Now, we look for $\H \in \Omega^{1}(E) \otimes \Omega^{2}(E)$, such that the connection $\hcD^{0}$ defined by 
\begin{equation} \label{eq_hcD0}
\hcD^{0}_{\psi}\psi' = \hcDL_{\psi}\psi' + g_{E}^{-1}\H(\psi,\psi',\cdot),
\end{equation}
for all $\psi,\psi' \in \Gamma(\gTM)$, becomes a Levi-Civita connection. It is easy to verify that the following $\H$ solves this problem:
\begin{equation} \label{eq_Hforcon}
\begin{split}
\H((X,\xi),(Y,\eta),(Z,\zeta)) = & \ \frac{1}{6} H'(g^{-1}(\xi),Y,g^{-1}(\zeta)) + \frac{1}{6} H'(g^{-1}(\xi),g^{-1}(\eta),Z) \\
& - \frac{1}{3} H'(X,Y,Z) - \frac{1}{3}	H'(X,g^{-1}(\eta),g^{-1}(\zeta)),
\end{split}
\end{equation}
for all $(X,\xi),(Y,\eta),(Z,\zeta) \in \Gamma(\gTM)$. Indeed, the requirement $\hcD^{0} \in \LC(\gTM,\G)$ is equivalent to the following list of conditions imposed on $\H$: 
\begin{align}
\label{eq_condonH1} 0 & = \H(\psi,\psi',\psi'') + \H(\psi,\psi'',\psi'), \\
\label{eq_condonH2} 0 & = \H(\psi,\psi',\tau(\psi'')) + \H(\psi,\psi'',\tau(\psi')), \\
\label{eq_condonH3} 0 & = \H(\psi,\psi',\psi'') + cyclic(\psi,\psi',\psi'') + H'(\rho(\psi),\rho(\psi'),\rho(\psi''),
\end{align}
for all $\psi,\psi',\psi'' \in \Gamma(\gTM)$. Note that $\tau(X,\xi) = (g^{-1}(\xi),g(X))$, for all $(X,\xi) \in \Gamma(\gTM)$. The first condition is forced by the compatibility $(\ref{eq_conncomp})$ of $\hcD^{0}$ with $\<\cdot,\cdot\>_{E}$, the second one is forced by the compatibility with $\G$, that is $\hcD^{0} \G = 0$, and the last one follows from the torsion-freeness condition. All three equations are easily verified to hold for (\ref{eq_Hforcon}). 
\begin{tvrz}
There exists a Levi-Civita connection $\hcD^{0} \in \LC(\gTM,\G)$, given by formulas (\ref{eq_hcDL}), (\ref{eq_hcD0}),  and (\ref{eq_Hforcon}). We call it the \textbf{minimal Levi-Civita connection} on $\gTM$. In the block form, it can be written as 
\begin{equation} \label{eq_minimalblock}
\hcD_{(X,\xi)} ^{0} = \bm{\cDL_{X} + \frac{1}{6} g^{-1}H'(g^{-1}(\xi),\star,\cdot)}{-\frac{1}{3}g^{-1}H'(X,g^{-1}(\star),\cdot)}{-\frac{1}{3}H'(X,\star,\cdot)}{\cDL_{X} + \frac{1}{6}H'(g^{-1}(\xi),g^{-1}(\star),\cdot)},
\end{equation}
for all $(X,\xi) \in \Gamma(\gTM)$. The symbols $\star$ denote the "inputs". In other words, one has 
\begin{equation}
\begin{split}
\hcD^{0}_{(X,\xi)}(Y,\eta) = \big( & \cDL_{X}Y + \frac{1}{6} g^{-1}H'(g^{-1}(\xi),Y,\cdot) - \frac{1}{3} g^{-1}H'(X,g^{-1}(\eta),\cdot), \\
& \cDL_{X}\eta - \frac{1}{3} H'(X,Y,\cdot) + \frac{1}{6} H'(g^{-1}(\xi),g^{-1}(\eta),\cdot) \big),
\end{split}
\end{equation}
for all $(X,\xi),(Y,\eta) \in \Gamma(\gTM)$. 
\end{tvrz}
Since we have now demonstrated that $\LC(\gTM,\G) \neq 0$, we know that there are infinitely many Levi-Civita connections on $\gTM$ thanks to Lemma \ref{lem_LCconnections}. 

In particular, one has $\rank(\LC_{0}(\gTM,\G)) = \frac{2}{3} n(n^{2}-1)$, where $n = \dim(M)$. Every other Levi-Civita connection $\hcD \in \LC(\gTM,\G)$ is now given a by formula 
\begin{equation} \label{eq_hcDashcD0andK}
\hcD_{\psi}\psi' = \hcD^{0}_{\psi}\psi' + g_{E}^{-1}\K(\psi,\psi',\cdot),
\end{equation}
where $\K$ has to satisfy (\ref{eq_condonH1}, \ref{eq_condonH2}) with $\H$ replaced by $\K$, together with (\ref{eq_Kcyclic}). For $E = \gTM$, the tensor $\K$ on $E$ can  conveniently be parametrized by a pair of ordinary tensors $J \in \vf{} \otimes \vf{2}$ and $W \in \df{1} \otimes \df{2}$. Namely, the most general $\K$ satisfying (\ref{eq_condonH1}, \ref{eq_condonH2}) can be written as 
\begin{equation} \label{eq_Ktensor}
\begin{split}
\K((X,\xi),(Y,\eta),(Z,\zeta)) = & \ W(g^{-1}(\xi),Y,Z) + W(X, g^{-1}(\eta),Z) \\
& + W(X,Y,g^{-1}(\zeta) + W(g^{-1}(\xi),g^{-1}(\eta),g^{-1}(\zeta)) \\
& - J(g(X),\eta,\zeta) - J(\xi,g(Y),\zeta) \\
& - J(\xi,\eta,g(Z)) - J(g(X),g(Y),g(Z)),
\end{split}
\end{equation}
for all $(X,\xi),(Y,\eta),(Z,\zeta) \in \Gamma(\gTM)$. Plugging into (\ref{eq_Kcyclic}) gives two independent condition on tensors $J$ and $W$, namely 
\begin{equation} \label{eq_JWcyclic}
J(\xi,\eta,\zeta) + cyclic(\xi,\eta,\zeta) = 0, \; \; W(X,Y,Z) + cyclic(X,Y,Z) = 0, 
\end{equation}
for all $\xi,\eta,\zeta \in \df{1}$ and $X,Y,Z \in \vf{}$. Equivalently, one has $J_{a} = 0$ and $W_{a} = 0$, with subscript $a$ denoting the complete skew-symmetrization of the tensor. This gives a complete description of the set $\LC(\gTM,\G)$. 

The main task of the following lines will be to find an expression for both scalar curvatures $\widehat{\RS}_{E}$ and $\widehat{\RS}_{\G}$ of the connection $\hcD$ in terms of the manifold Levi-Civita connection $\cDL$, the $3$-form $H'$ and the pair of tensors $W$ and $J$. In order to achieve this goal, we will employ the propositions \ref{tvrz_RS'asRSandK'} and \ref{tvrz_RSG'asRSGandK'}. In particular, this requires to calculate the scalar curvatures 
$\widehat{\RS}^{0}_{E}$ and $\widehat{\RS}^{0}_{\G}$ of the minimal connection $\hcD^{0}$. We formulate the result in the form of a proposition.
\begin{tvrz} \label{tvrz_hRS}
Let $\hcD^{0} \in \LC(\gTM,\G)$ be the minimal Levi-Civita connection (\ref{eq_minimalblock}). Then its scalar curvatures $\widehat{\RS}^{0}_{E}$ and $\widehat{\RS}^{0}_{\G}$ have the form
\begin{equation} \label{eq_hRS}
\widehat{\RS}^{0}_{E} = 0, \; \; \widehat{\RS}^{0}_{\G} = \RS(g) - \frac{1}{2} \<H',H'\>_{g},
\end{equation}
where $\RS(g)$ is the usual Ricci scalar curvature of the manifold Levi-Civita connection $\cDL$ corresponding to the metric $g$, and $\<\cdot,\cdot\>_{g}$ is the scalar product of differential forms using the Hodge duality operator $\ast_{g}$, that is for $\alpha,\beta \in \df{p}$, define
\begin{equation} 
\alpha \^ \ast_{g} \beta =: \<\alpha,\beta\>_{g} \cdot \vol_{g}, 
\end{equation}
where $\vol_{g} \in \df{n}$ is the canonical volume form on the metric manifold $(M,g)$. In index notation, one can write  $\<H',H'\>_{g} = \frac{1}{6} H'_{ijk} H'^{ijk},$
with indices being raised using $g^{-1}$. .
\end{tvrz}
\begin{proof}
We start from the expression (\ref{eq_hcD0}). One can repeat the calculation leading to the formula (\ref{eq_Ric'andRic}), except that $\hcDL$ is not torsion-free for nonzero $H'$, which leads also to terms containing the torsion $\widehat{T}^{LC}(\psi,\psi')$. Note that in this case, the partial trace $\H'$ vanishes, as $H' \in \df{3}$. One thus finds
\begin{equation} \label{eq_hRic0}
\begin{split}
\hRic{}^{0}(\psi,\psi') = & \ \hRic{}^{LC}(\psi,\psi') + \frac{1}{2} \{ (\hcDL_{\psi_{\lambda}}\H)(\psi,\psi',\psi^{\lambda}_{E}) + (\hcDL_{\psi_{\lambda}}\H)(\psi',\psi,\psi^{\lambda}_{E}) \\
& - \H(\psi', g_{E}^{-1}\H(\psi_{\lambda},\psi,\cdot), \psi^{\lambda}_{E}) - \H(\psi^{\lambda}_{E}, g_{E}^{-1}\H(\psi,\psi_{\lambda},\cdot),\psi') \\
& + \H(g_{E}^{-1}\H(\cdot,\psi_{\lambda},\psi'),\psi,\psi^{\lambda}_{E}) \\
& + \H(\widehat{T}^{LC}(\psi_{\lambda},\psi'),\psi,\psi^{\lambda}_{E}) + \H(\widehat{T}^{LC}(\psi,\psi^{\lambda}_{E}),\psi_{\lambda},\psi')\},
\end{split}
\end{equation}
where $\hRic{}^{LC}$ is the Ricci tensor corresponding to the connection $\hcDL$. 
All ingredients on the right-hand side can be easily calculated by hand from (\ref{eq_hcDL}) and (\ref{eq_Hforcon}). First, one finds
\begin{equation}
\hRic{}^{LC}(\psi,\psi') = \Ric{}^{LC}(\rho(\psi),\rho(\psi')), 
\end{equation}
for all $\psi,\psi' \in \Gamma(\gTM)$, where $\Ric{}^{LC}$ is the ordinary Ricci curvature tensor of the Levi-Civita connection $\cDL$ corresponding to the metric $g$. Next, one has 
\begin{equation}
\begin{split}
(\hcDL_{\psi_{\lambda}}\H')((X,\xi),(Y,\eta), \psi^{\lambda}_{E}) = & \ \frac{1}{6} (\cDL_{\psi_{k}}H')(g^{-1}(\xi), Y, \psi^{k}_{g}) \\
& - \frac{1}{3} (\cDL_{\psi_{k}}H')(X,g^{-1}(\eta),\psi^{k}_{g}),
\end{split}
\end{equation}
for all $(X,\xi),(Y,\eta) \in \Gamma(\gTM)$. Here $\{ \psi_{k} \}_{k=1}^{\dim(M)}$ is an arbitrary local frame on $TM$ and $\psi^{k}_{g} = g^{-1}(\psi^{k})$. Calculation of the terms in (\ref{eq_hRic0}) quadratic in $\H$ is straightforward by plugging in the expression (\ref{eq_Hforcon}), yet one has to be careful with signs and prefactors. We get
\begin{equation}
-\H((Y,\eta),g_{E}^{-1}\H(\psi_{\lambda},(X,\xi),\cdot),\psi^{\lambda}_{E}) = \frac{1}{18} \< \io_{g^{-1}(\xi)}H', \io_{g^{-1}(\eta)} H'\>_{g} - \frac{1}{9} \< \io_{X}H', \io_{Y}H'\>_{g}, 
\end{equation}
for all $(X,\xi), (Y,\eta) \in \Gamma(\gTM)$. The second quadratic term differs only in the order of sections $(X,\xi)$ and $(Y,\eta)$, and thus yields the same result. For third one, we have 
\begin{equation}
\begin{split}
\H(g_{E}^{-1}\H(\cdot,\psi_{\lambda},(Y,\eta)),(X,\xi), \psi^{\lambda}_{E}) = & \ \frac{2}{9} \<\io_{g^{-1}(\xi)}H', \io_{g^{-1}(\eta)}H'\>_{g} + \frac{2}{9} \<\io_{X}H',\io_{Y}H'\>_{g}.
\end{split}
\end{equation}
Finally, the first term containing the torsion $\widehat{T}^{LC}$ gives 
\begin{equation}
\H(\widehat{T}^{LC}(\psi_{\lambda},(Y,\eta)), (X,\xi), \psi^{\lambda}_{E}) = - \frac{1}{3} \< \io_{X}H',\io_{Y}H'\>_{g}
\end{equation}
and the second one gives the same result. 
We can now and collect the terms in (\ref{eq_hRic0}) and obtain the final expression
\begin{equation} \label{eq_hRic0final}
\begin{split}
\hRic{}^{0}((X,\xi),(Y,\eta)) = & \ \Ric{}^{LC}(X,Y) \\
& - \frac{1}{4} \{ (\cDL_{\psi_{k}}H')(X,g^{-1}(\eta),\psi^{k}_{g}) + (\cDL_{\psi_{k}}H')(Y,g^{-1}(\xi), \psi^{k}_{g}) \} \\
& - \frac{1}{3} \<\io_{X}H', \io_{Y}H'\>_{g} + \frac{1}{6} \<\io_{g^{-1}(\xi)}H', \io_{g^{-1}(\eta)}H'\>_{g}. 
\end{split}
\end{equation}
Calculation of the two scalar curvatures $\widehat{\RS}_{E}$ and $\widehat{\RS}_{\G}$ is now easy. 
\begin{equation}
\begin{split}
\widehat{\RS}_{E} = & \ \hRic{}^{0}((\psi_{k},0),(0,\psi^{k})) + \hRic{}^{0}((0,\psi^{k}),(\psi_{k},0)) \\
= & \ -\frac{1}{2} \{ ( \cDL_{\psi_{m}}H')(\psi_{k},\psi^{k}_{g}, \psi^{m}_{g}) \} = 0,
\end{split} 
\end{equation}
as $H'$ is completely skew-symmetric. Similarly, one gets 
\begin{equation}
\begin{split}
\widehat{\RS}_{\G} = & \ \Ric{}^{LC}(\psi_{k},\psi^{k}_{g}) + \hRic{}^{0}((\psi_{k},0),(\psi^{k}_{g},0)) + \hRic{}^{0}((0,\psi^{k}),(0,\psi_{k}^{g})) \\
= & \ \RS(g) - \frac{1}{2} \<H',H'\>_{g},
\end{split} 
\end{equation}
where we have used the notation $\psi_{k}^{g}$ for $g(\psi_{k}) \in \df{1}$. 
\end{proof}
Now, we can examine when the connection $\hcD^{0}$ is Ricci compatible with the generalized metric $\G$, as we have the explicit description of its Ricci tensor $\hRic{}^{0}$ in the proof above. This gives a neat equation restricting the background fields $(g,B)$. 
\begin{lemma} \label{lem_hRic0Riccomp}
The minimal connection $\hcD^{0}$ is Ricci compatible with $\G$, if and only if 
\begin{equation} \label{eq_hRic0Riccomp}
\Ric{}^{LC}(X,Y) - \frac{1}{2} (\delta_{g}H')(X,Y) - \frac{1}{2} \<\io_{X}H',\io_{Y}H'\>_{g} = 0,
\end{equation}
where $\delta_{g}: \Omega^{\bullet}(M) \rightarrow \Omega^{\bullet - 1}(M)$ is a codifferential induced by the metric on the manifold $(M,g)$. 
\end{lemma}
\begin{proof}
The generalized metric $\G$ gives a decomposition $\G = V_{+}^{0} \oplus V_{-}^{0}$, where $V_{\pm}^{0}$ are graphs of $\pm g \in \Hom(TM,T^{\ast}M)$. In other words, $\hcD^{0}$ is Ricci compatible, if and only if 
\begin{equation}
\hRic{}^{0}( (X,g(X)), (Y,-g(Y)) = 0, 
\end{equation}
for all $X,Y \in \vf{}$. Plugging into (\ref{eq_hRic0final}), one obtains (\ref{eq_hRic0Riccomp}). Recall that for any $H' \in \df{3}$, one can express the codifferential using the Levi-Civita connection as 
\begin{equation} (\delta_{g}H')(X,Y) = - (\cDL_{\psi_{k}}H')(\psi^{k}_{g},X,Y). 
\end{equation}
\end{proof}
The worst part of the calculation leading to the final expressions for scalar curvatures of $\hcD \in \LC(\gTM,\G)$ defined by (\ref{eq_hcDashcD0andK}) and (\ref{eq_Ktensor}) is successfully behind us. It only remains to apply the propositions \ref{tvrz_RS'asRSandK'} and \ref{tvrz_RSG'asRSGandK'}. We state the result as a theorem.
\begin{theorem} \label{thm_scalars} Consider $\gTM$ equipped with the $H'$-twisted Dorfman bracket and the generalized metric $\G$.
Let $\hcD \in \LC(\gTM,\G)$ be a general Levi-Civita connection on $\gTM$ parametrized by two tensors $J \in \vf{} \otimes \vf{2}$ and $W \in \df{1} \otimes \df{2}$ as in (\ref{eq_hcDashcD0andK}) and (\ref{eq_Ktensor}). Then 
\begin{align}
\label{eq_hRSEfinal} \widehat{\RS}_{E} & = - 4 \Div_{g}(J') + 8 \<J', W'\>, \\
\label{eq_hRSGfinal} \widehat{\RS}_{\G} & = \RS(g) - \frac{1}{2} \<H',H'\>_{g} + 4 \Div_{g}(W') - 4 \rVert W' \rVert_{g}^{2} - 4 \rVert J' \rVert_{g}^{2},
\end{align}
where $\Div_{g}$ is the covariant divergence using the manifold Levi-Civita connection $\cDL$, and $J'$ and $W'$ are partial traces defined by
\begin{equation}
J'(\zeta) = J'( \psi^{k}, \psi_{k}^{g}, \zeta), \; \; W'(Z) = W'( \psi_{k}, \psi^{k}_{g}, Z),
\end{equation}
for all $\zeta \in \df{1}$ and $Z \in \vf{}$. Again, $\{\psi_{k}\}_{k=1}^{\dim{M}}$ is an arbitrary local frame on $TM$,  $\{\psi^{k}\}_{k=1}^{\dim{M}}$ the dual one on $T^\ast M$ and $\psi^{k}_{g} := g^{-1}(\psi^{k})$, $\psi_{k}^{g} := g(\psi_{k})$. Finally, $\rVert \cdot \rVert_{g}$ denotes the usual (point-wise) norm of (co)vector fields induced by the metric $g$. 
\end{theorem}
\begin{proof}
To prove the formula for $\widehat{\RS}_{E}$, we employ the equation (\ref{eq_RS'asRSandK'}) together with (\ref{eq_hRS}). To do so, we must find the $1$-form $\K' \in \Omega^{1}(E)$, starting from (\ref{eq_Ktensor}). One finds
\begin{equation} \label{eq_KprimeasWandJ}
\K'(Z,\zeta) = 2 W'(Z) - 2 J'(\zeta),
\end{equation}
for all $(Z,\zeta) \in \Gamma(\gTM)$. For the covariant divergence $\Div_{\hcD^{0}}(\K')$, we find
\begin{equation}
\begin{split}
\Div_{\hcD^{0}}( \K') = & \ \rho(\psi_{\lambda}). \<\K', \psi^{\lambda}_{E}\> - \< \K', \hcD^{0}_{\psi_{\lambda}}( \psi^{\lambda}_{E}) \> \\
= & \ - 2 \cdot ( \psi_{k}.\<J', \psi^{k}\> - \<J', \cDL_{\psi_{k}}\psi^{k} \> ) \\
= & \ - 2 \cdot \Div_{g}(J'). 
\end{split}
\end{equation}
The norm $\rVert \K' \rVert^{2}_{E}$ can be calculated even more easily, giving
\begin{equation}
\rVert \K' \rVert^{2}_{E} \equiv \K'(\psi_{\lambda}) \cdot \K'(\psi^{\lambda}_{E}) = - 8 W'(\psi_{k}) \cdot J'(\psi^{k}) = - 8 \<J',W'\>. 
\end{equation}
Plugging these two partial results and (\ref{eq_hRS}) into (\ref{eq_RS'asRSandK'}) gives the equation (\ref{eq_hRSEfinal}). To prove (\ref{eq_hRSGfinal}), we will first employ the two isomorphisms $\fPsi_{\pm}^{0} \in \Hom(TM,V^{0}_{\pm})$, where $V^{0}_{\pm}$ are the two subbundles induced by $\G$. They are given by $\fPsi_{\pm}^{0}(X) = (X,\pm g(X))$, for all $X \in \vf{}$. We can now pull everything "downstairs" and view the connections $\cD^{\pm}$, fiber-wise metrics $\gm_{\pm}$ and tensors $\K_{\pm}$ as standard objects on $M$. We will denote them by the same letters. The ordinary manifold connections $\cD^{\pm}$ are determined by the equations
\begin{equation}
\hcD^{0}_{\fPsi^{0}_{+}(X)} \fPsi^{0}_{+}(Y) = \fPsi^{0}_{+}( \cD_{X}^{+}Y), \; \; \hcD^{0}_{\fPsi^{0}_{-}(X)} \fPsi^{0}_{-}(Y) = \fPsi^{0}_{-}( \cD_{X}^{-}Y),
\end{equation}
for all $X,Y \in \vf{}$. Plugging into (\ref{eq_minimalblock}), one finds 
\begin{equation}
\cD^{\pm}_{X}Y = \cDL_{X}Y \mp \frac{1}{6} g^{-1}H'(X,Y,\cdot),
\end{equation}
for all $X,Y \in \vf{}$. The two metrics $\gm_{\pm}$ are induced by 
\begin{equation}
\gm_{\pm}(X,Y) = \gm(\fPsi^{0}_{\pm}(X), \fPsi^{0}_{\pm}(Y)) = 2 g(X,Y),
\end{equation}
for all $X,Y \in \vf{}$, as  $\G = \BlockDiag(g,g^{-1})$ is block diagonal. Hence, $\gm_{\pm} = 2g$. Finally, from (\ref{eq_Ktensor}), one gets
\begin{equation}
\K_{\pm}(X,Y,Z) = \pm 4 W(X,Y,Z) - 4 J(g(X),g(Y),g(Z)). 
\end{equation}
Consequently, the partial traces $\K'_{\pm}$ defined using $\gm_{\pm}$ give
\begin{equation}
\K'_{\pm}(Z) \equiv \K_{\pm}( \psi_{k}, \gm_{\pm}^{-1}(\psi^{k}), Z) = \pm 2 W'(Z) - 2 J'(g(Z)),
\end{equation}
for all $Z \in \vf{}$. It remains to calculate the covariant divergences and norms. We find
\begin{equation}
\begin{split}
\Div_{\cD^{\pm}}(\K'_{\pm}) = & \ (\cD^{\pm}_{\psi_{k}} \K'_{\pm})( \gm_{\pm}^{-1}(\psi^{k})) = \frac{1}{2} ( \cD^{\pm}_{\psi_{k}} \K'_{\pm} ) (\psi^{k}_{g}) = \frac{1}{2} ( \pm \Div_{g}(W') - 2 \Div_{g}(J')) \\
= & \ \pm \Div_{g}(W') - \Div_{g}(J'). 
\end{split}
\end{equation}
Finally, we obtain 
\begin{equation}
\begin{split}
\rVert \K'_{\pm} \rVert^{2}_{\gm_{\pm}} = & \ \K'_{\pm}(\psi_{k}) \cdot \K'_{\pm}( \gm_{\pm}^{-1}(\psi^{k})) = \frac{1}{2} \K'_{\pm}(\psi_{k}) \cdot \K'_{\pm}(\psi^{k}_{g}) \\
= & \ 2 \rVert W' \rVert_{g}^{2} \mp 4 \<J',W'\> + 2 \rVert J' \rVert^{2}_{g}. 
\end{split}
\end{equation}
Plugging the last two results into (\ref{eq_RSG'asRSGandK'}) and using (\ref{eq_hRS}) gives precisely (\ref{eq_hRSGfinal}).
\end{proof}
It is quite interesting, yet not very surprising, that the resulting expressions (\ref{eq_hRSEfinal}, \ref{eq_hRSGfinal}) depend only on the scalars produced from partial traces $J'$ and $W'$. 

We can any time revert to the original untwisted connection $\cD$, which will have the same scalar curvatures, $\RS_{E} = \widehat{\RS}_{E}$ and $\RS_{\gm} = \widehat{\RS}_{\G}$. Before doing so, we will examine one very interesting property of $\hcD$, Namely, its Ricci compatibility with the generalized metric $\G$. The result is in fact surprisingly simple. 
\begin{theorem} \label{thm_Riccomp}
Let $\hcD \in \LC(\gTM,\G)$ be the most general Levi-Civita connection on $\gTM$ equipped with the $H'$-twisted Dorfman bracket and the generalized metric $\G$, parametrized by two tensors $J \in \vf{} \otimes \vf{2}$ and $W \in \df{1} \otimes \df{2}$, as in (\ref{eq_hcDashcD0andK}) and (\ref{eq_Ktensor}). Then $\hcD$ is Ricci compatible if and only if 
\begin{equation} \label{eq_hcDRiccomp1}
\begin{split}
0 = & \ \Ric{}^{LC}(X,Y) - \frac{1}{2} (\delta_{g}H')(X,Y) - \frac{1}{2} \<\io_{X}H',\io_{Y}H'\>_{g} \\
& + (\cDL_{X}W')(Y) + (\cDL_{Y}W')(X) + \< \io_{X}\io_{Y}H', W'\>_{g} \\
& + (\cDL_{X}J')(g(Y)) - (\cDL_{Y}J')(g(X)),
\end{split}
\end{equation}
for all $X,Y \in \vf{}$. This is equivalent to a pair of equations
\begin{align}
\label{eq_hcDRiccomp2} 0 = & \ \Ric{}^{LC}(X,Y) - \frac{1}{2}\<\io_{X}H',\io_{Y}H'\>_{g} + (\cDL_{X}W')(Y) + (\cDL_{Y}W')(X), \\
\label{eq_hcDRiccomp3} 0 = & \ (\cDL_{X}J')(g(Y)) - (\cDL_{Y}J')(g(X)) - \frac{1}{2} (\delta_{g}H')(X,Y) + \< \io_{X}\io_{Y}H', W'\>_{g},
\end{align}
imposed on all vector fields $X,Y \in \vf{}$ on $M$. 
\end{theorem}
\begin{proof}
This follows almost immediately from (\ref{eq_Ric'andRic}) and from the proof of Lemma \ref{lem_hRic0Riccomp}. Indeed, (\ref{eq_Ric'andRic}) implies the equality
\begin{equation}
\begin{split}
\hRic(\fPsi_{+}^{0}(X), \fPsi_{-}^{0}(Y)) = & \ \hRic{}^{0}(\fPsi_{+}^{0}(X), \fPsi_{-}^{0}(Y)) \\
& + \frac{1}{2} \{(\hcD^{0}_{\fPsi_{+}^{0}(X)} \K')(\fPsi^{0}_{-}(Y)) + (\hcD^{0}_{\fPsi^{0}_{-}(Y)} \K')(\fPsi^{0}_{+}(X)) \},
\end{split}
\end{equation}
for all $X,Y \in \vf{}$. All other terms containing $\K$ give zero, as $\K$ vanishes whenever evaluated on any pair of sections $(\psi,\psi')$, where $\psi \in \Gamma(V_{+})$ and $\psi' \in \Gamma(V_{-})$. Moreover, the compatibility of $\hcD^{0}$ with $\G$ implies that also $\hcD_{\psi}^{0}\K$ has the same property. Next, it follows from the proof of Lemma \ref{lem_hRic0Riccomp} that $\hRic{}^{0}(\fPsi^{0}_{+}(X),\fPsi^{0}_{-}(Y))$ is exactly the left-hand side of (\ref{eq_hRic0Riccomp}). It remains to evaluate the remaining two terms using (\ref{eq_minimalblock}) and (\ref{eq_KprimeasWandJ}). One has
\begin{align}
(\hcD_{\fPsi_{+}^{0}(X)} \K')(\fPsi_{-}^{0}(Y)) = & \ 2 \{ (\cD_{X}W')(Y) + (\cD_{X}J')(g(Y)) \} \\
& - H'(X,Y,\psi_{k}) \cdot \{ W'(\psi^{k}_{g}) + J'(\psi^{k}) \}, \nonumber \\
(\hcD_{\fPsi_{-}^{0}(Y)} \K')(\fPsi_{+}^{0}(X)) = & \ 2 \{ (\cD_{Y}W')(X) - (\cD_{Y}J')(g(X)) \} \\
& - H'(X,Y,\psi_{k}) \cdot \{ W'(\psi^{k}_{g}) - J'(\psi^{k}) \}. \nonumber 
\end{align}
Summing up the both sides and dividing by two, we obtain the remaining terms in (\ref{eq_hcDRiccomp1}). Equations (\ref{eq_hcDRiccomp2}) and (\ref{eq_hcDRiccomp3}) are the symmetric and skew-symmetric part of (\ref{eq_hcDRiccomp1}) in $(X,Y)$, respectively. 
\end{proof}
The fact that the Ricci compatibility again depends only on the partial traces $J'$ and $W'$ is considerably more interesting  than in the case of above equations (\ref{eq_hRSEfinal}) and 
(\ref{eq_hRSGfinal}). Again, the original connection $\cD \in \LC(\gTM,\gm)$ is Ricci compatible with $\gm$ if and only if $\hcD$ is Ricci compatible with $\G$. This follows from Lemma \ref{lem_riccompcovariance}. Thus, it would be very convenient to have an interpretation for $J'$ and $W'$ directly in terms of the connection $\cD$.  This is, without any doubt, possible for the vector field $J' \in \vf{}$. 
\begin{tvrz} \label{tvrz_J'ascharvf}
Let $\hcD \in \LC(\gTM,\G)$ be the Levi-Civita connection defined by (\ref{eq_hcDashcD0andK}) and (\ref{eq_Ktensor}). Let $X_{\hcD} \in \vf{}$ be the characteristic vector field of $\hcD$, defined by (\ref{eq_charvf}). Then 
\begin{equation} \label{eq_J'ascharvf}
J' = \frac{1}{2} X_{\hcD}. 
\end{equation}
In particular, if $\cD \in \LC(\gTM,\gm)$ is the connection related to $\hcD$ by (\ref{eq_conntwist}), we have $J' = \frac{1}{2} X_{\cD}$, where $X_{\cD}$ is the characteristic vector field of $\cD$. 
\end{tvrz}
\begin{proof}
For $\hcD$ given by (\ref{eq_hcDashcD0andK}), we have, directly from the definition of the covariant divergence: 
\begin{equation}
\Div_{\hcD}(\psi) = \Div_{\hcD^{0}}(\psi) - \K'(\psi). 
\end{equation}
Now, note that $\D{f} = (0,df)$ and that, as it is easy to check using (\ref{eq_minimalblock}), also $\Div_{\hcD^{0}}(\D{f}) = 0$. Plugging in from (\ref{eq_KprimeasWandJ}), we thus find the relation 
\begin{equation}
\< X_{\hcD}, df \> \equiv \Div_{\hcD}(\D{f}) = - \K'(0,df) = 2 J'(df). 
\end{equation}
This proves the assertion (\ref{eq_J'ascharvf}). Remaining statements follow from the fact that characteristic vector fields are invariant under Courant algebroid isomorphisms, see Lemma \ref{lem_XcDisavf}.
\end{proof}
For $1$-form $W'$, there also exists a description directly in terms of the connection $\cD$. However, it is significantly more cumbersome when compared to the above expression for $J'$. 
\begin{tvrz} \label{tvrz_K'asjinak}
Let $\hcD \in \LC(\gTM,\G)$ be the Levi-Civita connection defined by (\ref{eq_hcDashcD0andK}) and (\ref{eq_Ktensor}). Let $V_{\hcD} \in \vf{} \otimes \vf{2}$ be the tensor field defined by (\ref{eq_VcD}). Let $h_{\G}$ be the symmetric bilinear form (\ref{eq_hgm}) associated to the generalized metric $\G$. For $\gTM$ being an exact Courant algebroid, $h_{\G} > 0$ is a positive definite fiber-wise metric on $T^{\ast}M$. Then
\begin{equation} 
W(X,Y,Z) = V_{\hcD}( h_{\G}^{-1}(X), h_{\G}^{-1}(Y), h_{\G}^{-1}(Z)),
\end{equation}
for all $X,Y,Z \in \vf{}$. Thus, one has $W'(Z) = V_{\hcD}( \psi^{k}, h_{\G}^{-1}(\psi_{k}), h_{\G}^{-1}(Z))$. Finally, one can express $W$ and $W'$ in the same way using $V_{\cD}$ and $h_{\gm}$ associated to the connection $\cD$ and the generalized metric $\gm$. 
\end{tvrz}
\begin{proof}
In the proof of Lemma \ref{lem_hgcovariant}, we have shown that $h_{\G} = g^{-1}$. Also, we have $\rho^{\ast}(\xi) = (0,\xi)$, for all $\xi \in \df{1}$. It is now straightforward to use (\ref{eq_minimalblock}), (\ref{eq_hcDashcD0andK}) and (\ref{eq_Ktensor}) to show that
\begin{equation}
V_{\hcD}(\xi,\eta,\zeta) = W(g^{-1}(\xi),g^{-1}(\eta),g^{-1}(\zeta)) = W(h_{\G}(\xi),h_{\G}(\eta),h_{\G}(\zeta)). 
\end{equation}
The formula for $W'$ then follows. The rest follows from Lemma \ref{lem_hgcovariant}, which asserts that $h_{\gm} = h_{\G}$. Finally, similarly to the discussion below (\ref{eq_VcD}), $V_{\hcD} = V_{\cD}$. 
\end{proof}
\begin{rem}
It might seem that the propositions \ref{tvrz_J'ascharvf} and \ref{tvrz_K'asjinak} are trivial reformulations fitting this particular example. However, it can happen that one works with an exact Courant algebroid $E$ which is isomorphic to $\gTM$ only after a choice of a suitable isotropic splitting. The above propositions allow us to calculate the fields $J'$ and $W'$ directly, working with the original Courant algebroid structure. 
\end{rem}
\section{Equations of motion}\label{sec_EOM}
This section should be a pinnacle of this little piece of writing. We will discuss how the structures introduced in the previous sections provide a geometrical framework for equations arising from string theory. We leave out any overall constants from the picture. An interested reader could find everything in detail e.g. in the classical books \cite{GSW}, \cite{polchinski1998string} or in the lecture notes \cite{Tong:2009np}. 

One considers the bosonic string moving in the target manifold $M$ endowed with background fields $(g,B,\phi)$, where $g > 0$ is a metric, $B \in \df{2}$ and $\phi \in \cif$ is a scalar field called the \textbf{dilaton}. A crucial property for a consistent quantization of such theory is the so called Weyl invariance. After a non-trivial calculation, one can show that this leads to a necessary condition, the vanishing of the so called beta functions. In physics literature, these are usually written in the index notation, where $\{x^{\mu}\}_{\mu=1}^{\dim{M}}$ are some local coordinates on $M$:\footnote{assuming $\dim{M}=26$}
\begin{align}
\label{eq_betag} \beta(g)_{\mu \nu} & = \Ric{}^{LC}_{\mu \nu} - \frac{1}{4} H'_{\mu \lambda \kappa} {H'_{\nu}}^{\lambda \kappa} + 2 (\partial_{\mu} \phi)_{;\nu} , \\
\label{eq_betaB} \beta(B)_{\mu \nu} & = - \frac{1}{2} {H'^{\lambda}}_{\mu \nu ; \lambda} + {H'_{\mu \nu}}^{\lambda} (\partial_{\lambda}\phi), \\
\label{eq_betaphi} \beta(\phi) & = \RS(g) - \frac{1}{12} H'_{\mu \nu \lambda} H'^{\mu \nu \lambda} + 4 (\partial^{\mu} \phi)_{;\mu} - 4 (\partial_{\mu}\phi) (\partial^{\mu}\phi).
\end{align}
Here $H' = dB$. Note that $\beta(g)$ and $\beta(B)$ form a symmetric and skew-symmetric tensor on $M$, respectively. The Weyl invariance imposes the condition $\beta(g)_{\mu \nu} = \beta(B)_{\mu \nu} = \beta(\phi) = 0$. Note that this equations are not independent, as $\beta(g)_{\mu \nu} = 0$ implies $\RS(g) = \frac{1}{4} H'_{\mu \nu \lambda} H'^{\mu \nu \lambda} - 2 (\partial^{\mu} \phi)_{;\mu}$, which can be used to eliminate $\RS(g)$ in $\beta(\phi)$. For example in \cite{Tong:2009np}, they thus use a different beta function:
\begin{equation}
\beta'(\phi) = - \frac{1}{2} (\partial^{\mu}\phi)_{;\mu} + (\partial_{\mu}\phi)(\partial^{\mu}\phi) - \frac{1}{24} H'_{\mu \nu \lambda} H'^{\mu \nu \lambda}. 
\end{equation}
They are related as $\beta'(\phi) = - \frac{1}{4} (\beta(\phi) - \beta(g)_{\mu \nu} g^{\mu \nu} )$. Replacing $\beta(\phi)$ by $\beta'(\phi)$ in the Weyl invariance condition clearly gives the equivalent set of equations. Definitions (\ref{eq_betag} - \ref{eq_betaphi}) can be rewritten in the index-free notation. One finds
\begin{align}
(\beta(g))(X,Y) & = \Ric{}^{LC}(X,Y) - \frac{1}{2} \<\io_{X}H',\io_{Y}H'\>_{g} + (\cDL_{X}d\phi)(Y) + (\cDL_{Y}d\phi)(X), \\
(\beta(B))(X,Y) & = \frac{1}{2} (\delta_{g}H')(X,Y) + H'(X,Y, \cD^{g} \phi) = \frac{1}{2} e^{2\phi} \delta_{g}(e^{-2\phi} H')(X,Y), \\
\beta(\phi) & = \RS(g) - \frac{1}{2} \<H',H'\>_{g} + 4 \Delta_{g}(\phi) - 4 \rVert \cD^{g}\phi \rVert^{2}_{g},
\end{align}
for all $X,Y \in \vf{}$, where $\delta_{g}$ is a codifferential, $\Delta_{g}$ is the Laplace-Bertrami operator induced by $g$, and $\cD^{g}\phi$ denotes a gradient of the function $\phi$. Amazingly, the set of equations $\beta(g) = \beta(B) = \beta(\phi) = 0$ can be obtained as a set of equations of motion of a classical field theory, the \textbf{low-energy effective action} of the bosonic string. From now on, we can consider a more general definition, where $H' = H + dB$ for some closed $3$-form $H \in \df{3}$. 

\begin{tvrz} \label{tvrz_effaction}
Let $M$ be a manifold, and let $S$ be an action functional defined by
\begin{equation} \label{eq_effaction}
S[g,B,\phi] = \int_{M} e^{-2\phi} \{ \RS(g) - \frac{1}{2} \<H',H'\>_{g} + 4 \rVert \cD^{g} \phi \rVert^{2}_{g} \} \cdot \vol_{g},
\end{equation}
where $g$ is a metric, $B \in \df{2}$, and $\phi \in \cif$. Then $(g,B,\phi)$ is an extremal of the functional $S$, if and only if $\beta(g) = \beta(B) = \beta(\phi) = 0$. We say that $(g,B,\phi)$ \textbf{satisfy the equations of motion} of a field theory given by the action $S$. 
\end{tvrz}
\begin{proof}
We do not provide a full proof of this assertion, as it requires a considerable amount of work to prove how exactly $S$ changes under variations. We greatly recommend the book \cite{fecko2006differential} for a detailed treatment of variational problems coming from the geometrical formulation of classical field theories. First, let $\vartheta \in \cif$ be any function vanishing on the boundary $\partial M$, and let $\epsilon > 0$ be small real parameter. One finds 
\begin{equation}
S[g,B,\phi'] = S[g,B,\phi] - 2 \epsilon \int_{M} e^{-2\phi} \{ \beta(\phi) \cdot \vartheta \} \cdot \vol_{g} + o(\epsilon^{2}).
\end{equation}
Next, let $g' = g + \epsilon h$, where $h$ is arbitrary symmetric bilinear form vanishing on $\partial M$. Then\footnote{This particular "then" requires a calculation several pages long.}
\begin{equation}
S[g',B,\phi] = S[g,B,\phi] - \epsilon \int_{M} e^{-2\phi} h^{\mu \nu} \{ \beta(g)_{\mu \nu} - \frac{1}{2} \beta(\phi) g_{\mu \nu} \} \cdot \vol_{g} + o(\epsilon^{2}),
\end{equation}
where in $h^{\mu \nu}$, the indices are raised using $g$.  Finally, let $B' = B + \epsilon C$ for any $C \in \df{2}$ vanishing on $\partial M$. One obtains the expression 
\begin{equation}
S[g,B',\phi] = S[g,B,\phi] - 2 \epsilon \int_{M} e^{-2\phi} \<C, \beta(B)\>_{g} \cdot \vol_{g} + o(\epsilon^{2}).
\end{equation}
It follows that $(g,B,\phi)$ is an extremal of $S$, if an only if all terms containing the first power of $\epsilon$ in all three above expressions vanish for all $\theta$, $h$ and $C$, which is in turn equivalent to $\beta(\phi) = 0$, $\beta(g)_{\mu \nu} - \frac{1}{2} \beta(\phi) g_{\mu \nu} = 0$, and $\beta(B) = 0$. This finishes the proof. 
\end{proof}

As the reader familiar with the previous sections already suspects, we can now formulate the equations of motion in terms of the Courant algebroid connections. We already have  prepared all the tools required to prove the corresponding statement.  It is now an easy consequence of the previous theorems. Let us note that the idea to describe the vanishing of beta functions using the generalizations of connections and their curvatures is not new \cite{2013arXiv1304.4294G}. In double field theory \cite{Hull:2009mi}, this is one of central ideas, see in particular \cite{Hohm:2012mf} or some survey papers \cite{Aldazabal:2013sca, Hohm:2013bwa}. Let us now formulate and prove the central theorem relating beta functions to Courant algebroid connections. 

\begin{theorem} \label{thm_central}
Let $(E,\rho,\<\cdot,\cdot\>_{E},[\cdot,\cdot]_{E})$ be the Courant algebroid on $E = \gTM \equiv TM \oplus T^{\ast}M$, described in Example \ref{ex_Hdorfman}. Let $\gm$ be a generalized metric on $\gTM$ corresponding to a pair $(g,B)$, where $g > 0$ is a Riemannian metric, and $B \in \df{2}$. 

Let $\cD \in \LC(\gTM,\gm)$ be a Levi-Civita connection with vanishing characteristic vector field (\ref{eq_charvf}), that is $X_{\cD} = 0$. Let $\phi \in \cif$ be a scalar field, and assume that 
\begin{equation} \label{eq_vcDisdphi}
V_{\cD}( \psi^{k}, h_{\gm}^{-1}(\psi_{k}), h_{\gm}^{-1}(Z)) = (d\phi)(Z), 
\end{equation}
for all $Z \in \vf{}$, where $V_{\cD} \in \vf{} \otimes \vf{2}$ is defined by (\ref{eq_VcD}) and $h_{\gm}$ is the fiber-wise metric defined by (\ref{eq_hgm}). Here, $\{\psi_{k}\}_{k=1}^{\dim(M)}$ is an arbitrary local frame on $TM$. Let $\dim(M) > 1$. 

Then $(g,B,\phi)$ satisfies the equations of motion, $\beta(g) = \beta(B) = \beta(\phi) = 0$, if and only if its Ricci scalar curvature $\RS_{\gm}$ vanishes, and it is Ricci compatible with $\gm$, that is 
\begin{equation}
\RS_{\gm} = 0, \ \ \Ric(V_{+},V_{-}) = 0, 
\end{equation}
where $V_{\pm} \subseteq \gTM$ are two vector subbundles induced by generalized metric $\gm$. 
\end{theorem}
\begin{proof}
First, one has to show that some connection $\cD$ satisfying the assumptions of the theorem exists. Using the propositions \ref{tvrz_J'ascharvf} and \ref{tvrz_K'asjinak} and the classification of Levi-Civita connections in the previous section, this is equivalent to finding $J$ and $W$, such that $J' = 0$ and $W' = d\phi$. Clearly, we can choose $J = 0$ and we can define $W$ to have the form 
\begin{equation} \label{eq_Wchoice}
W(X,Y,Z) = (1/\dim(M)) \{ g(X,Y) \<d\phi,Z\> - g(X,Z) \< d\phi, Y\> \},
\end{equation}
for all $X,Y,Z \in \vf{}$. See that $W_{a} = 0$. Taking the partial trace, we find $W' = d\phi$. With this choice of $W$ and $J$, we obtain the connection $\cD$ of required properties. 

The rest of the proof is simple. First, one employs Theorem \ref{thm_scalars}, in particular the equation (\ref{eq_hRSGfinal}). We find that for $\cD$ satisfying the assumptions, we have $\RS_{\gm} = \widehat{\RS}_{\G} = \beta(\phi)$. We know, cf. Theorem \ref{thm_Riccomp}, that $\cD$ is Ricci compatible with $\gm$ iff the corresponding connection (\ref{eq_conntwist}) $\hcD$ is Ricci compatible with $\G$. Using Theorem \ref{thm_Riccomp}, we see that $\hcD$ is Ricci compatible with $\G$ if and only if (\ref{eq_hcDRiccomp2}, \ref{eq_hcDRiccomp3}) hold. But for $J' = 0$ and $W' = d\phi$, the equation (\ref{eq_hcDRiccomp2}) becomes $\beta(g) = 0$ and (\ref{eq_hcDRiccomp3}) becomes $\beta(B) = 0$. In fact, note that 
\begin{equation}
\Ric(\fPsi^{0}_{+}(X),\fPsi^{0}_{-}(Y)) = (\beta(g))(X,Y) - (\beta(B))(X,Y),
\end{equation}
for all $X,Y \in \vf{}$. This finishes the proof. 
\end{proof}
Thus we have shown that vanishing beta functions or equivalently the equations of motion of the low-energy effective action can be fully reformulated in terms of Courant algebroid connections. Background fields $(g,B)$ come from the generalized metric $\gm$, whereas the dilaton field $\phi$ enters through the connection and equation (\ref{eq_vcDisdphi}). At this moment, there is no clear geometrical interpretation of the conditions imposed on the connection. 
\begin{rem}
Note that for any $\cD \in \LC(\gTM,\gm)$ with $X_{\cD} = 0$, one has $\RS_{E} = 0$. However, these two conditions are not equivalent. 
\end{rem}
\section{Application: Background independent gauge} \label{sec_BIG}
To demonstrate the power of the geometrical interpretation of equations of motion given by Theorem \ref{thm_central}, we derive an equivalence of the low-energy effective action (\ref{eq_effaction}) with a theory whose fields consist of a metric, twisted Poisson bivector and a dilaton. In fact, this bivector is assumed to be non-degenerate, hence, according to \cite{Blumenhagen:2012nt}, we should call this action a symplectic gravity. Before wading through the actual proof, we have to discuss the following generalization of Section \ref{sec_exact}. 

\begin{rem} \label{rem_ADorfmanatd}

Assume that $(A,a,[\cdot,\cdot]_{A})$ is a Lie algebroid over a base manifold $M$. Let $\Li{}^{A}$ and $d^{A}$ be operators on the exterior algebra $\Omega^{\bullet}(A)$ induced by the bracket $[\cdot,\cdot]_{A}$ as in the discussion below (\ref{eq_rhoishom}). 

One can form a Courant algebroid structure on $E = A \oplus A^{\ast}$ similar to the Example \ref{ex_Hdorfman}. The fiber-wise metric $\<\cdot,\cdot\>_{E}$ is defined using the same formula as (\ref{eq_canpairing}), the anchor takes the form $\rho = a \circ pr_{A}$, and one sets $[\cdot,\cdot]_{E}$ to be the \textbf{$H$-twisted $A$-Dorfman bracket}:
\begin{equation} \label{eq_AHDorfman}
[(\varphi,\vartheta),(\varphi',\vartheta')]_{E} =  ([\varphi,\varphi]_{A}, \Li{\varphi}^{A} \vartheta' - \io_{\varphi'} (d^{A}\varphi) - H(\varphi,\varphi',\cdot)),
\end{equation}
for all $(\varphi,\vartheta), (\varphi',\vartheta') \in \Gamma(E)$, where $H \in \Omega^{3}(A)$ satisfies $d^{A}H = 0$. It is easy to check that $(E,\rho,\<\cdot,\cdot\>_{E},[\cdot,\cdot]_{E})$ forms a Courant algebroid. It is regular (or transitive) if and only if the original Lie algebroid is. It is exact if and only if $a$ is a vector bundle isomorphism. It follows that one can now redo everything derived in Section \ref{sec_exact}, with $\gm$ corresponding to a pair $(g_{A},B_{A})$, where $g_{A}$ is now a positive definite fiber-wise metric on $A$, and $B_{A} \in \Omega^{2}(A)$. We get $H' = H + d^{A}B_{A}$. 

Instead of the manifold Levi-Civita connection, now one uses a Lie algebroid Levi-Civita connection $\cDL: \Gamma(A) \times \Gamma(A) \rightarrow \Gamma(A)$ with respect to $g_{A}$, given by the usual formula 
\begin{equation} \label{eq_cDLA}
\cDL_{\varphi}\varphi' = \frac{1}{2}\{ [\varphi,\varphi']_{A} + g^{-1}_{A}( \Li{\varphi}^{A}(g_{A}(\varphi')) + \io_{\varphi'}(d^{A}(g_{A}(\varphi))) \},
\end{equation}
for all $\varphi,\varphi' \in \Gamma(A)$. This connection is torsion-free in the usual sense, that is 
\begin{equation}
\cDL_{\varphi} \varphi' - \cDL_{\varphi'}\varphi - [\varphi,\varphi']_{A} = 0,
\end{equation}
for all $\varphi,\varphi' \in \Gamma(A)$. Also, one can define the corresponding curvature operator $R^{LC}$ using the usual formula
\begin{equation}
R^{LC}(\varphi,\varphi')\varphi'' = \cDL_{\varphi} \cDL_{\varphi'}\varphi'' - \cDL_{\varphi'} \cDL_{\varphi}\varphi'' - \cDL_{[\varphi,\varphi']_{A}}\varphi'' ,
\end{equation}
for all $\varphi,\varphi',\varphi'' \in \Gamma(A)$, and define the corresponding Ricci tensor $\Ric{}^{LC} \in \T_{2}^{0}(A)$ and the Ricci scalar $\RS(g_{A}) \in \cif$. All in all, $\cDL$ and its induced quantities $R^{LC}$, $\Ric{}^{LC}$ and $\RS(g_{A})$ enjoy the same properties as the ones for usual Levi-Civita connection. 

We can now, word for word, replicate the whole Section \ref{sec_exact} for $E = A \oplus A^{\ast}$, hence proving that $\LC(A \oplus A^{\ast},\gm) \neq 0$, where all connections are uniquely parametrized by a pair of tensors $J_{A} \in \mathfrak{X}(A) \otimes \mathfrak{X}^{2}(A)$ and $W_{A} \in \Omega^{1}(A) \otimes \Omega^{2}(A)$, such that $(J_{A})_{a} = (W_{A})_{a} = 0$. The only exceptions are Propositions \ref{tvrz_J'ascharvf} and \ref{tvrz_K'asjinak}, as the anchor $\rho = a \circ pr_{A}$ is not in general surjective. Instead, we get the relation
\begin{equation} \label{eq_J'AasXcD}
\< \xi, X_{\cD} \> = J'_{A}( a^{T}(\xi)), 
\end{equation}
for all $\xi \in \df{1}$. Moreover, the bilinear form $h_{\gm}$ is just the pullback of the fiber-wise metric $g^{-1}_{A}$ on $A^{\ast}$, that is $h_{\gm}(\xi,\eta) = g^{-1}_{A}(a^{T}(\xi),a^{T}(\eta))$. It is non-degenerate if and only if $a$ is surjective. Consequently, in general,  we only get the expression 
\begin{equation}
V_{\cD}(\xi,\eta,\zeta) = W_{A}(g_{A}^{-1}a^{T}(\xi), g_{A}^{-1}a^{T}(\eta), g_{A}^{-1}a^{T}(\zeta)),
\end{equation} \label{eq_vcDasWA}
for all $\xi,\eta,\zeta \in \df{1}$. Naturally, there is no direct relation for the (partial) trace. 
\end{rem}
Now, let us turn our attention back to physics. Assume that $(g,B)$ is background on the manifold $M$ equivalently described by the generalized metric $\gm$. Assume that $B$ is an almost symplectic $2$-form, that is $B \in \Hom(TM,T^{\ast}M)$ is a vector bundle isomorphism. Let $\theta \in \vf{2}$ be a bivector on $M$ defined to be its inverse, $\theta = B^{-1}$. It is a well-known fact that such $\theta$ is a $dB$-twisted Poisson tensor on $M$, defined in \cite{Severa:2001qm}. For any $\theta \in \vf{2}$, the Schouten-Nijenhuis bracket $[\theta,\theta]_{S} \in \vf{3}$ of $\theta$ with itself can be written as 
\begin{equation} \label{eq_schouten}
\frac{1}{2} [\theta,\theta]_{S}(\xi,\eta,\cdot) = [\theta(\xi),\theta(\eta)] - \theta( \Li{\theta(\xi)}\eta - \io_{\theta(\eta)}d\xi),
\end{equation}
for all $\xi, \eta \in \df{1}$. See e.g. Kosmann-Schwarzbach \cite{Kosmann1996} for a more detailed discussion. Plugging in $\xi = B(X)$ and $\eta = B(Y)$, the right-hand side gives $-dB(X,Y,\theta(\cdot))$ and we obtain 
\begin{equation} \label{eq_Poissontwisted}
\frac{1}{2} [\theta,\theta]_{S}(\xi,\eta,\zeta) = - dB(\theta(\xi),\theta(\eta),\theta(\zeta)),
\end{equation}
for all $\xi,\eta,\zeta \in \df{1}$. This is the defining equation of a \textbf{dB-twisted Poisson manifold}. Equivalently, for any $H \in \df{3}$ we may define a bracket $[\cdot,\cdot]_{\theta}^{H}: \df{1} \times \df{1} \rightarrow \df{1}$ as
\begin{equation}
[\xi,\eta]^{H}_{\theta} = \Li{\theta(\xi)} \eta - \io_{\theta(\eta)} d\xi + H(\theta(\xi),\theta(\eta),\theta(\cdot)),
\end{equation}
for all $\xi,\eta \in \df{1}$. Sometimes, it is called an \textbf{$H$-twisted Koszul bracket}. The twisted Jacobi identity (\ref{eq_Poissontwisted}) can be now reformulated as follows. 
\begin{tvrz} \label{tvrz_Koszul}
Any $\theta \in \vf{2}$ satisfies (\ref{eq_Poissontwisted}) if and only if the triple $(T^{\ast}M, \theta, [\cdot,\cdot]_{\theta}^{dB})$ defines a Lie algebroid on $T^{\ast}M$. 
\end{tvrz}
\begin{proof}
Leibniz rule (\ref{eq_lebnizrule}) and the skew-symmetry of $[\cdot,\cdot]_{\theta}^{dB}$ are clearly valid for any $\theta \in \vf{2}$. Jacobi identity (\ref{eq_leibnizidentity}) can be, using (\ref{eq_schouten}), readily recast into the form (\ref{eq_Poissontwisted}). 
\end{proof}

The existence of a twisted Poisson manifold is now used to define a new kind of orthogonal transformation on $\gTM$, and one examines the corresponding "twisted" Courant algebroid. It turns out that the result fits precisely into the framework of Remark \ref{rem_ADorfmanatd} for $A = T^{\ast}M$ endowed with the Lie algebroid structure described in Proposition \ref{tvrz_Koszul}. 

\begin{tvrz}
Let $\theta \in \vf{2}$ satisfy (\ref{eq_Poissontwisted}). Define $\F_{\theta} \in \End(\gTM)$ as 
\begin{equation}
\F_{\theta}(X,\xi) = (\theta(\xi), \xi - B(X))
\end{equation}
for all $(X,\xi) \in \gTM$. Let $(\gTM,\rho,\<\cdot,\cdot\>_{E},[\cdot,\cdot]_{E})$ be the Courant algebroid structure defined by $H$-twisted Dorfman bracket described in Example \ref{ex_Hdorfman}. Define a new bracket $[\cdot,\cdot]_{E}^{\theta}$ and $\rho_{\theta} \in \Hom(\gTM,TM)$ as a  "twist" of the original structure by the map $\F_{\theta}$: 
\begin{equation}
[\psi,\psi']_{E}^{\theta} = \F_{\theta}^{-1}[\F_{\theta}(\psi), \F_{\theta}(\psi')]_{E}, \; \; \rho_{\theta}(\psi) = \rho( \F_{\theta}(\psi)), 
\end{equation}
for all $\psi,\psi' \in \Gamma(\gTM)$. 

Then $(\gTM, \rho_{\theta}, \<\cdot,\cdot\>_{E}, [\cdot,\cdot]_{E}^{\theta})$ is a Courant algebroid. Moreover, putting $A = (T^{\ast}M, \theta,[\cdot,\cdot]_{\theta}^{dB})$, its bracket is the $H'_{\theta}$-twisted $A$-Dorfman bracket, where $H'_{\theta} \in \vf{3} = \Omega^{3}(A)$ is 
\begin{equation} \label{eq_H'theta}
H'_{\theta}(\xi,\eta,\zeta) = H'(\theta(\xi),\theta(\eta),\theta(\zeta)), 
\end{equation}
for all $\xi,\eta,\zeta \in \Gamma(A) \equiv \df{1}$. We have $H' = H + dB$. Equivalently, thanks to (\ref{eq_Poissontwisted}), we can write
\begin{equation}
H'_{\theta} = - \frac{1}{2}[\theta,\theta]_{S} + H_{\theta},
\end{equation}
where $H_{\theta}$ is defined using $H$ and $\theta$ similarly as in (\ref{eq_H'theta}). 
\end{tvrz}
\begin{proof}
It is clear that $(\gTM,\rho_{\theta},\<\cdot,\cdot\>_{E},[\cdot,\cdot]^{\theta}_{E})$ forms a Courant algebroid, as both the bracket $[\cdot,\cdot]_{E}^{\theta}$ and anchor $\rho_{\theta}$ are defined in order to make $\F_{\theta}$ into a Courant algebroid isomorphism. One only has to show that $\F_{\theta}$ is orthogonal with respect to $\<\cdot,\cdot\>_{E}$: 
\begin{equation}
\begin{split}
\<\F_{\theta}(X,\xi), \F_{\theta}(Y,\eta)\>_{E} = & \ \< (\theta(\xi),\xi - B(X)), (\theta(\eta), \eta - B(Y) \>_{E} \\
= & \ \theta(\xi,\eta) + \theta(\eta,\xi) - \< B(X),\theta(\eta)\> - \<B(Y),\theta(\xi)\> \\
= & \ \< X, B(\theta(\eta))\> + \<Y, B(\theta(\xi))\> = \<\eta,X\> + \<\xi,Y\> \\
= & \ \<(X,\xi),(Y,\eta)\>_{E}. 
\end{split}
\end{equation}
The inverse $\F_{\theta}^{-1}$ has an explicit form $\F_{\theta}^{-1}(X,\xi) = (X - \theta(\xi), B(X))$, for all $(X,\xi) \in \Gamma(\gTM)$. In the remainder of this proof, we will write sections of $\gTM$ suggestively in the opposite order, that is $(\xi,X) \in \Gamma(\gTM)$ for $X \in \vf{}$ and $\xi \in \df{1}$. By explicit calculation, one finds
\begin{equation}
[(\xi,0),(\eta,0)]_{E}^{\theta} = \big( \theta^{-1}[\theta(\xi),\theta(\eta)], [\theta(\xi),\theta(\eta)] - \theta( \Li{\theta(\xi)}(\eta) - \io_{\theta(\eta)} d\xi ) - H_{\theta}(\xi,\eta,\zeta) \big),
\end{equation}
for all $\xi,\eta \in \df{1}$. Now, as $(T^{\ast}M,\theta,[\cdot,\cdot]_{\theta}^{dB})$ forms a Lie algebroid, $\theta$ is a bracket morphism (\ref{eq_rhoishom}) and thus $[\theta(\xi),\theta(\eta)] = \theta([\xi,\eta]_{\theta}^{dB})$. Moreover, one can use (\ref{eq_schouten}) and (\ref{eq_Poissontwisted}) to rewrite the terms in the second component of the right-hand side. We find the expression
\begin{equation}
[(\xi,0),(\eta,0)]^{\theta}_{E} = ( [\xi,\eta]_{\theta}^{dB}, -H'_{\theta}(\xi,\eta,\cdot)),
\end{equation}
for all $\xi,\eta \in \df{1}$. This is in agreement with (\ref{eq_AHDorfman}). Finally, for all $X,Y \in \vf{}$, we have
\begin{equation}
[(0,X),(0,Y)]_{E}^{\theta} = 0.
\end{equation}
This is again the correct value for (\ref{eq_AHDorfman}). In fact, these two special cases complete the proof as all the mixed terms are already uniquely determined by the axiom (\ref{eq_invarianceznovu}). 
\end{proof}

Now, consider a new generalized metric $\gm_{\theta} = \F^{T}_{\theta} \gm \F$. Using (\ref{eq_minimalblock}), one finds
\begin{equation}
\gm_{\theta} = \F_{\theta}^{T} \gm \F_{\theta} = \bm{G^{-1}}{0}{0}{G}, 
\end{equation}
where $G$ is the Riemannian metric on $M$ given by $G = -Bg^{-1}B$. This is called the \textbf{background-independent gauge}, the name going back to Seiberg and Witten in \cite{Seiberg:1999vs}. In the context of Remark \ref{rem_ADorfmanatd}, we obtain $g_{A} = G^{-1}$. 

Before stating the main theorem, let us introduce some notation for objects on Lie algebroid $(T^{\ast}M,\theta,[\cdot,\cdot]_{\theta}^{dB})$. We denote its Lie algebroid differential as $d_{\theta}$. The Levi-Civita connection (\ref{eq_cDLA}) corresponding to the fiber-wise metric $G^{-1}$ will be denoted\footnote{To save some space, everybody should call it a Levi-Citheta connection.} as $\cD^{L\theta}$, and its Ricci tensor and Ricci curvature as $\Ric^{\theta}$ and $\RS^{\theta}(G^{-1})$, respectively. Correspondingly, $\Delta_{\theta}$ is the Laplace-Bertrami operator defined using the connection $\cD^{L\theta}$, that is $\Delta_{\theta}(\phi) = (\cD^{L\theta}_{\psi^{k}}(d_{\theta}\phi))(G(\psi_{k}))$. Finally, by $\<\cdot,\cdot\>_{G}$ we mean a direct analogue of the scalar product of forms, for example 
\begin{equation}
\<H'_{\theta},H'_{\theta}\>_{G} = \frac{1}{6} H'_{\theta}( \psi^{k}, \psi^{q}, \psi^{l}) \cdot H'_{\theta}( G(\psi_{k}), G(\psi_{q}), G(\psi_{l})). 
\end{equation}
We are now able to quickly prove the following theorem
\begin{theorem}
Let $\gm$ be a generalized metric on $\gTM$, corresponding to a pair $(g,B)$, and let $\phi \in \cif$ be a scalar field. Let $\cD \in \LC(\gTM,\gm)$ be any connection satisfying the assumptions of Theorem \ref{thm_central}. Let $\cD^{\theta} \in \LC(\gTM, \gm_{\theta})$ be the Levi-Civita connection on Courant algebroid $(\gTM,\rho_{\theta},\<\cdot,\cdot\>_{E},[\cdot,\cdot]_{\theta})$ with respect to $\gm_{\theta}$ defined by 
\begin{equation} \label{eq_cDthetaascD}
\F_{\theta}(\cD^{\theta}_{\psi}\psi') = \cD_{\F_{\theta}(\psi)} \F_{\theta}(\psi'), 
\end{equation}
for all $\psi,\psi' \in \Gamma(\gTM)$. Let $\RS_{\gm_{\theta}}$ be its Ricci scalar corresponding to $\gm_{\theta}$. Then $\cD^{\theta}$ is Ricci compatible with $\gm_{\theta}$ and $\RS_{\gm_{\theta}} = 0$, if and only if 
\begin{align}
\label{eq_symplectic1} \RS^{\theta}(G^{-1}) - \frac{1}{2} \<H'_{\theta},H'_{\theta}\>_{G} + 4 \Delta_{\theta}(\phi) - 4 \rVert d_{\theta} \phi \rVert^{2}_{G} &=  0, \\
\label{eq_symplectic2} \Ric^{\theta} (\xi,\eta) - \frac{1}{2} \<\io_{\xi} H'_{\theta}, \io_{\eta} H'_{\theta} \>_{G} + (\cD^{L\theta}_{\xi} (d_{\theta}\phi))(\eta) + (\cD^{L\theta}_{\eta} (d_{\theta}\phi))(\xi) &= 0 , \\
\label{eq_symplectic3} H'_{\theta}(\xi,\eta,G(d_{\theta}\phi)) - \frac{1}{2} (\cD^{L\theta}_{\psi^{k}} H'_{\theta})(G(\psi_{k}), \xi,\eta) & = 0, 
\end{align}
for all $\xi,\eta \in \df{1}$.
\end{theorem}
\begin{proof}
The statement follows from the discussion in Remark \ref{rem_ADorfmanatd}. We have to check that $J'_{A} = 0$ and $K'_{A} = d_{\theta}\phi$. The rest is implied by generalizations of Theorems \ref{thm_scalars} and \ref{thm_Riccomp} valid for $A$-Dorfman brackets. By assumption, we have $X_{\cD} = 0$. As $\F_{\theta}$ is a Courant algebroid isomorphism and $\cD^{\theta}$ is defined by (\ref{eq_cDthetaascD}), it follows from Lemma \ref{lem_XcDisavf} that $X_{\cD^{\theta}} = X_{\cD}$. From (\ref{eq_J'AasXcD}), we get 
\begin{equation}
J'_{A}(\theta(\xi)) = - \<\xi, X^{\theta}_{\cD} \> = - \<\xi, X_{\cD}\> = 0,
\end{equation}
for all $\xi \in \df{1}$. As $\theta$ is invertible, we have $J'_{A} = 0$. Next, by assumption, we have $V_{\cD}(\psi^{k}, \psi_{k}^{g}, g(Z)) = (d\phi)(Z)$. Again, one has $V_{\cD} = V_{\cD^{\theta}}$. Plugging into (\ref{eq_vcDasWA}), we get 
\begin{equation}
V_{\cD}(\xi,\eta,\zeta) = - W_{A}( G\theta(\xi), G\theta(\eta), G\theta(\zeta)), 
\end{equation}
for all $\xi,\eta,\zeta \in \df{1}$. Finally, combining these two, we find
\begin{equation}
\begin{split}
(d_{\theta}\phi)(\zeta) = & \ (d\phi)(\theta(\zeta)) = V_{\cD}(\psi^{k},\psi_{k}^{g}, g\theta(\zeta)) = -W_{A}(G\theta(\psi^{k}), G\theta g(\psi_{k}), G\theta g \theta(\zeta)) \\
= & \ W_{A}(G(\psi'_{k}), \psi'^{k}, \zeta) \equiv W'_{A}(\zeta),
\end{split}
\end{equation}
for all $\zeta \in \df{1}$. We have redefined the local frame, setting $\psi'_{k} \:= \theta(\psi^{k})$, and two times used the relation $G = -Bg^{-1}B$. 
\end{proof}
\begin{cor} \label{cor_equations}
The background fields $(g,B,\phi)$ satisfy the equations of motions of the field theory given by (\ref{eq_effaction}) if and only if the equations (\ref{eq_symplectic1} - \ref{eq_symplectic3}) hold. 
\end{cor}
\begin{proof}
We know that the Ricci compatibility is transferred via the Courant algebroid isomorphisms, see Lemma \ref{lem_riccompcovariance}. Similarly, one has $\RS_{\gm} = \RS_{\gm_{\theta}}$. The rest follows from the previous theorem combined with the assertions of Theorem \ref{thm_central}. 
\end{proof}

The equations (\ref{eq_symplectic1} - \ref{eq_symplectic3}) can again be obtained as extremality conditions for a classical field theory action. It is the one called in \cite{Blumenhagen:2012nt} the \textbf{symplectic gravity}. For simplicity, consider now only the case $H = 0$, and define a $3$-vector $\Theta$ as 
\begin{equation}
\Theta \equiv H'_{\theta} = - \frac{1}{2}[\theta,\theta]_{S}. 
\end{equation}
Note that $\Theta$ is usually called the $R$-flux. Moreover, one can consider an alternative volume form $\vol_{G}^{\theta} \in \df{n}$, defined by an equation
\begin{equation}
\vol^{\theta}_{G}( \theta(\xi_{1}), \dots, \theta(\xi_{n})) = \vol_{G^{-1}}(\xi_{1}, \dots, \xi_{n}), 
\end{equation}
for all $\xi_{1},\dots,\xi_{n} \in \df{1}$, where $\vol_{G^{-1}} \in \vf{n}$ is the top degree multivector field constructed similarly as the usual metric volume form. In any positively oriented set of local coordinates $(x^{1},\dots,x^{n})$, this gives 
\begin{equation}
\vol^{\theta}_{G} = |G^{-1}|^{\frac{1}{2}} |\theta|^{-1} dx^{1} \^ \cdots \^ dx^{n}
\end{equation}
The symplectic gravity is given by the action functional:
\begin{equation} \label{eq_actionSG}
S'[G,\theta,\phi] = \int_{M} e^{-2\phi} \{ \RS^{\theta}(G^{-1}) - \frac{1}{2} \<\Theta,\Theta\>_{G} + 4 \rVert d_{\theta} \phi \rVert^{2}_{G} \} \cdot \vol^{\theta}_{G}.
\end{equation}
The proof of showing that equations of motion for $(G,\theta,\phi)$ given by the requirement $\delta S' = 0$ are exactly the equations (\ref{eq_symplectic1} - \ref{eq_symplectic3}) is quite non-trivial, see \cite{Blumenhagen:2012nt} for more detailed comments. 

The corollary (\ref{cor_equations}) thus provides a simple proof of the equivalence of the full set of equations of motion of the low-energy effective action of the bosonic string (\ref{eq_effaction}) and the symplectic gravity defined by (\ref{eq_actionSG}). Thus, this equivalence can be, from the geometrical point of view, interpreted as a particular example of a Courant algebroid isomorphism. However, it is still quite mysterious why the variation of the actions leads precisely onto the vanishing of the Ricci scalar and the Ricci compatibility condition. 
\section{Application: Kaluza-Klein type reduction} \label{sec_KKR}
In this section, we will be very brief, see \cite{Jurco:2015bfs} for a detailed discussion. However, note that in the cited paper, we did not use the present definition (\ref{eq_Rtensor}) of the Riemann tensor, which posed many technical difficulties. In particular, the Ricci tensor $\Ric$ used there was pretty ugly, and we did not have the condition of the Ricci compatibility with the generalized metric. 

The main idea is to consider an exact Courant algebroid $\gTP$ over a principal $G$-bundle $\pi: P \rightarrow M$, which can  under certain conditions \cite{Bursztyn2007726} be reduced to a Courant algebroid $E'$ over $M$. In particular, one can also reduce a generalized metric and to some extent also Levi-Civita connections. In light of Theorem \ref{thm_central}, we seek for some relations between field theories targeted in $P$ and $M$, providing a Kaluza-Klein type of reduction. 

One chooses an arbitrary but fixed principal connection $A \in \Omega^{1}(P,\g)$, and assumes that there exists $H_{0} \in \df{3}$, such that
\begin{equation} \label{eq_HonTP}
H = \pi^{\ast}(H_{0}) + \frac{1}{2} CS_{3}(A)
\end{equation}
is a closed form on $P$. Here  $F \in \Omega^{2}(M,\g_{P})$ is a curvature $2$-form valued in the sections of the adjoint bundle $\g_{P}$, and $\<\cdot,\cdot\>_{\g}$ is a Killing form on $\g$. In other words, we assume that the first Pontriyagin class of $P$ vanishes. Let $[\cdot,\cdot]_{E}$ be an $H$-twisted Dorfman bracket on $\gTP$. We consider $G$ to be a compact and semisimple Lie group. 

Next, define a $\R$-linear map $\Re: \g \rightarrow \Gamma(E)$ as $\Re(x) = \#{x} - \frac{1}{2}\<A,x\>_{\g}$ for every $x \in \g$. It is defined so that $x \blacktriangleright \psi = [\Re(x),\psi]_{E}$, for all $\psi \in \Gamma(E)$, becomes a Lie algebra action integrating to the usual right action of $G$ on $\gTP$. Now, consider a subbundle $K^{\perp}$ whose sections are
\begin{equation}
\Gamma(K^{\perp}) = \{ \psi \in \Gamma(\gTP) \; | \; \<\psi,\Re(x)\>_{E} = 0 \text{ for all $x \in \g$ } \}. 
\end{equation}
It follows that $K^{\perp} \subseteq \gTP$ is $G$-invariant. Moreover, the corresponding $\cif$-module $\Gamma_{G}(K^{\perp})$ of $G$-invariant sections is involutive under $[\cdot,\cdot]_{E}$. This allows one to obtain a reduced Courant algebroid structure on a vector bundle $E'$ over $M$, defined by $\Gamma(E') = \Gamma_{G}(K^{\perp})$. Given a splitting induced by connection $A$, this vector bundle is isomorphic to $TM \oplus \g_{P} \oplus T^{\ast}M$. See \cite{Baraglia:2013wua} for details and \cite{Bursztyn2007726} for a more general construction. The explicit form of the bracket on $E'$ can be also found in our paper \cite{Jurco:2015bfs}. 

One can thus expect that under some conditions, both generalized metric and  some corresponding Levi-Civita connections can be reduced as well. Let $\gm$ be a generalized metric on $\gTP$, and let $\tau$ be the corresponding involution. We impose
\begin{equation}
\tau([\Re(x), \psi]_{E}) = [\Re(x), \tau(\psi)]_{E}, \; \; \tau(K^{\perp}) \subseteq K^{\perp},
\end{equation}
for all $\psi \in \Gamma(E)$ and $x \in \g$. The first of the two conditions forces the corresponding $(g,B)$ to be $G$-invariant tensors on $P$. In particular, they can be decomposed with respect to the splitting $\Gamma_{G}(TP) \cong TM \oplus\g_{P}$ given be the connection. The second forces these block forms to be
\begin{equation} \label{eq_gandBrelevant}
g = \bm{1}{\vartheta^{T}}{0}{1} \bm{g_{0}}{0}{0}{-\frac{1}{2}c} \bm{1}{0}{\vartheta}{1}, \; \; B =\bm{B_{0}}{\frac{1}{2}\vartheta^{T}c}{-\frac{1}{2}c\vartheta}{0},
\end{equation}
where $g_{0}$ is a Riemannian metric on the manifold $M$, $B_{0} \in \df{2}$ and $\vartheta \in \Omega^{1}(M,\g_{P})$. These are exactly the three background fields parametrizing any generalized Riemann metric $\gm'$ on $E'$. By definition, this $\gm'$ is obtained by a restriction of $\gm$ onto $\Gamma_{G}(K^{\perp}) \cong E'$. 

Now, let $\cD \in LC(E,\gm)$. One assumes that for $\psi,\psi'\in \Gamma_{G}(E)$, we have $\cD_{\psi}\psi' \in \Gamma_{G}(E)$, and if moreover $\psi,\psi' \in \Gamma_{G}(K^{\perp})$, then $\cD_{\psi}\psi' \in \Gamma_{G}(K^{\perp})$. We can then define the Courant algebroid connection $\cD'$ on $E'$ by restriction of $\cD$ onto $\Gamma_{G}(K^{\perp})$. Every Courant algebroid connection on $E'$ can be obtained in this way. If $\cD \in \LC(E,\gm)$, then $\cD' \in \LC(E',\gm')$. However, unlike in the case of generalized metric, $\cD$ satisfying the above conditions is not uniquely determined by $\cD'$. 

In \cite{Jurco:2015bfs}, we took an opposite approach.  One can choose a convenient connection $\cD' \in \LC(E',\gm')$ and extend it to a connection $\cD \in \LC(E,\gm)$ which reduces back to $\cD'$. In particular, one can find the following relation between their respective scalar curvatures:
\begin{equation} \label{eq_RSrelations}
\RS_{\gm} = \RS'_{\gm'} \circ \pi + \frac{1}{6} \dim{\g}, \; \; \RS_{E} = \RS'_{E'} \circ \pi + \frac{1}{6} \dim{\g}. 
\end{equation}
The scalar curvatures $\RS'_{\gm'}$ and $\RS'_{E'}$ can be calculated to give 
\begin{align}
\label{eq_RSonM1} \RS'_{\gm'} = & \ \RS(g_{0}) + \frac{1}{2} \< \! \< F' , F' \> \! \> - \frac{1}{2} \<H'_{0},H'_{0}\>_{g_{0}} + 4 \Delta_{g_{0}}(\phi_{0}) - 4 \rVert d\phi_{0} \rVert^{2}_{g_{0}} + \frac{1}{6} \dim{\g} , \\ 
\RS'_{E'} = & \ - \frac{1}{6} \dim{\g},
\end{align}
where $F' = F + D{\vartheta} + \frac{1}{2} [\vartheta \^ \vartheta]_{\g}$, $H'_{0} = H_{0} + dB_{0} - \frac{1}{2}\tC_{3}(\vartheta) - \<F \^ \vartheta \>_{\g}$, and $\tC_{3}(\vartheta) = \< D{\vartheta} \^ \vartheta \>_{\g} + \frac{1}{3} \< [\vartheta \^ \vartheta]_{\g} \^ \vartheta \>_{\g}$. Here, $D$ denotes the exterior covariant derivative induced by $A$ on $\Omega^{\bullet}(M,\g_{P})$ and $\< \! \< \cdot,\cdot\> \! \>$ is a combination of the $p$-form product $\<\cdot,\cdot\>_{g_{0}}$ with the fiber-wise metric $\<\cdot,\cdot\>_{\g}$ on $\g_{P}$. In view of Theorem \ref{thm_central}, one can now show that $X_{\cD} = (X_{\cD'})^{h} = 0$ and $W' = \pi^{\ast}(d\phi_{0})$. 

This proves that the condition $\RS_{\gm} = 0$ corresponds to the equation of motion for dilaton $\phi = \pi^{\ast}(\phi_{0})$ together with $(g,B)$ given by (\ref{eq_gandBrelevant}) and $H$ defined by (\ref{eq_HonTP}). From (\ref{eq_RSrelations}), this is equivalent to $\RS'_{\gm'} + (1/6) \cdot \dim{\g} = 0$. In turn, from (\ref{eq_RSonM1}), this is equivalent to the equation of motion for dilaton $\phi_{0}$ of the field theory action on $M$ given by 
\begin{equation} \label{eq_hetaction}
S_{0} = \int_{M} e^{-2\phi_{0}} \{ \RS(g_{0}) + \frac{1}{2} \< \! \< F', F' \> \! \> - \frac{1}{2} \<H'_{0},H'_{0}\>_{g} + 4 \rVert d\phi_{0} \rVert^{2}_{g_{0}} + \frac{1}{3} \dim{\g} \} \cdot \vol_{g}. 
\end{equation}
In other words, the reduction of Courant algebroids working together with the Levi-Civita connections leads to a proposal for a Kaluza-Klein type of reduction. The above action can now be compared to the bosonic part of the low energy effective action for the heterotic string, see e.g. \cite{polchinski1998string}. The equations of motion arising from this theory can be related to systems of partial differential equations called Strominger systems. A detailed introduction and a discussion of their relations with string theory and generalized geometry is given e.g. in \cite{Garcia-Fernandez:2016azr}. We also recommend this paper for a more complete list of references.
\begin{rem}
At the beginning of this section, we have assumed that the first Pontryagin class of the principal $G$-bundle $\pi: P \rightarrow M$ vanishes. However, this is not at all important for the results of this section. Instead of a Courant algebroid, we could simply work with pre-Courant algebroids, see Remark \ref{rem_precourant}. In particular, $H_{0} \in \df{3}$ could have been completely arbitrary. 
\end{rem}
\section*{Acknowledgement}
It is a pleasure to thank Peter Bouwknegt, Urs Schreiber and Satoshi Watamura for helpful discussions.
The research of B.J. was supported by grant GA\v CR P201/12/G028 and in part by the Action MP1405 QSPACE from COST. He would like to thank the Tohoku Forum for Creativity for hospitality. 
The research of J.V. was supported by RVO: 67985840, he would like to thank the Max Planck Institute for Mathematics in Bonn for hospitality. 
\bibliography{bib} 
\end{document}